\documentclass[pra,aps,floatfix,amsmath,superscriptaddress,twocolumn,longbibliography,nofootinbib]{revtex4-2}
\usepackage{wrapfig, blindtext}
\usepackage[most]{tcolorbox}

\usepackage{amssymb,enumerate}
\usepackage{graphicx}
\usepackage{graphics}
\usepackage{amsmath}
\usepackage{amsthm,bbm}
\usepackage{color}
\usepackage{dsfont}
\usepackage{hyperref}

\usepackage{array}

\usepackage{amsmath}
\usepackage{tikz-cd}
\usepackage{empheq}
\usepackage{graphicx}

\usepackage{anyfontsize}
\usepackage{hyperref}
\usepackage[capitalise]{cleveref}
\usepackage{nameref}
\usepackage{xpatch}

\usepackage{qcircuit}

\usepackage{xcolor}
\hypersetup{
    colorlinks,
    linkcolor={red!50!black},
    citecolor={blue!50!black},
    urlcolor={blue!80!black}
}

\usepackage{amsmath}
\usepackage{tikz-cd}
\usepackage{empheq}

\usepackage{enumitem}

\usepackage{graphicx}

\usepackage{hyperref}
\usepackage[capitalise]{cleveref}
\usepackage{youngtab}

\usepackage{tikz}
\usepackage{mathtools}

\renewcommand{\d}{\mathrm{d}}
\renewcommand{\i}{\mathrm{i}}

\newcommand{\Tr}{\operatorname{Tr}}

\newcommand{\<}{\langle}
\renewcommand{\>}{\rangle}

\renewcommand{\i}{\mathrm{i}}
\renewcommand{\d}{\mathrm{d}}
\renewcommand{\i}{\mathrm{i}}

\newtheorem{thm}{Theorem}
\newenvironment{thmbis}[1]
  {\renewcommand{\thethm}{\ref{#1}$'$}%
   \addtocounter{thm}{-1}%
   \begin{thm}}
  {\end{thm}}

\usepackage{lipsum}
\usepackage{lmodern}
\usepackage{tcolorbox}
\makeatletter
\xpatchcmd{\@ssect@ltx}{\@xsect}{\protected@edef\@currentlabelname{#8}\@xsect}{}{}
\xpatchcmd{\@sect@ltx}{\@xsect}{\protected@edef\@currentlabelname{#8}\@xsect}{}{}
\makeatother
\usepackage{youngtab}

\usepackage{tikz}

\usepackage{amsfonts}
\usepackage{graphicx,graphics,epsfig,times,bm,bbm,amssymb,amsmath,amsfonts,mathrsfs}
\usepackage[normalem]{ulem}
\usepackage{setspace}
\usepackage{subfigure}
\usepackage{dsfont}
\usepackage{braket}
\usepackage{upgreek }
\usepackage{tikz}
\usepackage{natbib}
\usepackage{chngcntr}

\newtheorem{theorem}{Theorem}

\newtheorem{corollary}[theorem]{Corollary}

\newtheorem{lemma}[theorem]{Lemma}

\newtheorem{proposition}[theorem]{Proposition}

\newcommand{\bes} {\begin{subequations}}
\newcommand{\ees} {\end{subequations}}
\newcommand{\bea} {\begin{eqnarray}}
\newcommand{\eea} {\end{eqnarray}}
\newcommand{\be} {\begin{equation}}
\newcommand{\ee} {\end{equation}}

\def\>{\rangle}
\def\<{\langle}
\def\Tr{\textrm{Tr}}

\newcommand{\ignore}[1]{}

\renewcommand{\i}{\mathrm{i}}

\begin{document}	
	\title{Theory of Quantum Circuits with Abelian Symmetries}

	\author{Iman Marvian}
\affiliation{Duke Quantum Center, Departments of Physics \& Electrical and Computer Engineering, Duke University, Durham, North Carolina}

	\begin{abstract}
Quantum circuits with gates (local unitaries) respecting a global symmetry have broad applications in quantum information science and related fields, such as condensed matter theory and quantum thermodynamics. However, despite their widespread use, fundamental properties of such circuits are not well-understood. Recently, it was found that generic unitaries respecting a global symmetry cannot be realized, even approximately, using gates that respect the same symmetry. This observation raises important open questions: What unitary transformations can be realized  with k-local gates that respect a global symmetry? In other words, in the presence of a global symmetry,  how does the locality of interactions constrain the possible time evolution of a composite system? In this work, we address these questions for the case of Abelian (commutative) symmetries and develop constructive methods for synthesizing circuits with such symmetries. Remarkably, as a corollary, we find that, while the locality of interactions still imposes additional constraints on realizable unitaries, certain restrictions observed in the case of non-Abelian symmetries do not apply to circuits with Abelian symmetries.  For instance, in circuits with a general non-Abelian symmetry such as SU($d$), the unitary realized in a subspace with one irreducible representation (charge) of the symmetry dictates the realized unitaries in multiple other sectors with inequivalent representations of the symmetry. Furthermore, in certain sectors, rather than all unitaries respecting the symmetry, the realizable unitaries are the symplectic or orthogonal subgroups of this group. We prove that none of these restrictions appears in the case of Abelian symmetries. This result suggests that global non-Abelian symmetries may affect the thermalization of quantum systems in ways not possible under Abelian symmetries.
 \end{abstract}

		\maketitle

\section{Introduction}

The quantum circuit model \cite{deutsch1985quantum, deutsch1989quantum, divincenzo1995two, lloyd1995almost, Barenco:95a, NielsenAndChuang} was originally developed in the context of quantum computing, inspired by  (classical) digital logic circuits. This model provides a framework for formulating quantum algorithms and can be directly translated into instructions that can be implemented on quantum computers. Over time, the applications of this model have extended well beyond quantum computing. Quantum circuits have now become a standard framework for describing the dynamics and phases of many-body quantum systems (See, e.g., \cite{chen2010local, chen2011classification, khemani2018operator, roberts2017chaos, cotler2017chaos}). Hence, given the wide range of motivations and applications of quantum circuits, understanding how the standard framework of quantum circuits should be modified in the presence of symmetries is also crucial and has broad applications.

For instance, suppose certain native gates on a quantum computer respect the U(1) symmetry corresponding to rotations around the z-axis or respect  its $\mathbb{Z}_2$ subgroup corresponding to $\pi$ rotation around z. What sort of unitary transformations can be realized with such gates? Recently, it was observed  that in the presence of symmetries, the locality of unitaries in the circuit further restricts the set of realizable unitaries \cite{marvian2022restrictions, alhambra2022forbidden}. That is, there is no fixed $k$ such that symmetric unitaries that are $k$-local (i.e., act non-trivially on at most $k$ qudits) generate all symmetric unitaries on systems with $n>k$ qudits. This is in sharp contrast to the well-known universality of 2-local unitaries in the absence of symmetries \cite{divincenzo1995two, lloyd1995almost}.

What is the nature of the constraints imposed by the locality of the gates, and how do they depend on the properties of the symmetry group? In particular, are there any fundamental distinctions between the restrictions imposed by the locality in the case of Abelian and non-Abelian symmetries?
This question can be equivalently formulated in the language of the dynamics of composite systems with local Hamiltonians. The presence of symmetries implies certain standard conservation laws, a fact often referred to as Noether's theorem \cite{noether1918nachrichten, noether1971invariant}. These standard conservation laws apply to both Abelian and non-Abelian symmetries. For systems with local interactions, does the presence of non-Abelian global symmetries impose any additional constraints (beyond the standard conservation laws) that do not appear in the case of Abelian symmetries? Answering this question will be a step toward understanding the special features of non-Abelian symmetries in the dynamics and thermalization of quantum systems, which have recently attracted significant attention. (See, e.g., \cite{majidy2023noncommuting, halpern2016microcanonical, majidy2023non, guryanova2016thermodynamics, halpern2020noncommuting,  majidy2023critical, lostaglio2017thermodynamic}).




\subsection*{Overview of main results}

In this work we develop the theory of quantum circuits with Abelian symmetries.  
Our main results in theorems \ref{Thm3}, \ref*{Thm4}, and  \ref{Thm2} characterize the group of unitaries generated by  $k$-local unitaries that respect an  Abelian global symmetry. The proofs of these results are constructive. Indeed, we develop new techniques for synthesizing circuits with Abelian symmetries.  
Interestingly, our results imply that certain  constraints that have been previously observed in the case of non-Abelian symmetries \cite{Marvian2021qudit}, never appear in circuits with  Abelian symmetries.  In the following, we present an informal overview of these results. 

To explain this, first recall that the unitary representation of any global  symmetry group $G$ decomposes the total Hilbert space of $n$ qudits, as
$(\mathbb{C}^d)^{\otimes n}\cong \bigoplus_{\mu} \mathcal{H}_\mu$,   where the summation is over inequivalent  irreducible representations (irreps) of group $G$,   and  $ \mathcal{H}_\mu$ is the subspace of states with irrep $\mu$.  Then, a unitary transformation  on this system respects the Abelian symmetry $G$, i.e., is $G$-invariant  if, and only if, it is block-diagonal with respect to this decomposition, which means it conserves the charge (irrep) associated to the symmetry.
However, it turns out that a  general unitary  with this property cannot be realized with $k$-local unitaries that respect the same symmetry. In particular, as  shown in \cite{marvian2022restrictions},  
\vspace{-1mm}
\begin{enumerate}[label=\textbf{\Roman*}.]
  \setcounter{enumi}{0}
\item 
Locality imposes certain constraints on  the relative phases between sectors with different charges $\{\mathcal{H}_\mu\}$ (See theorem \ref{Thm-1} and the discussion below it).  
\end{enumerate}
Motivated by the presence of these constraints, in this paper we propose a weaker notion of universality, which will be called semi-universality\footnote{This name refers to universality on the semi-simple part of the Lie algebra of symmetric Hamiltonians.}. Namely, $k$-local symmetric unitaries are called semi-universal  if they generate all symmetric unitaries,  up to constraints on the relative phases between sectors with different charges. More precisely, if for any desired symmetric unitary $V$  there exists a set of phases $\phi_\mu\in(-\pi,\pi]$, such that $V\sum_\mu e^{\i\phi_\mu}\Pi_\mu$ is realizable by    $k$-local symmetric unitaries, we say  $k$-local symmetric unitaries  are semi-universal, where $\Pi_\mu$ is the projector to $\mathcal{H}_\mu$, the sector with charge $\mu$.
 (See Sec. \ref{def:semi} for the formal definitions).

\begin{table*}[htb]
  \centering
  \renewcommand{\arraystretch}{2}
  \setlength{\tabcolsep}{4pt}
  \begin{tabular}{ | m{3cm} | m{5cm}| m{4cm} | m{4cm} | } 
  \hline
\center{\textbf{Symmetry}}\\ $ $\\  & \center{\vspace{-2mm}\textbf{      Representation of symmetry on each local subsystem}\\ } 
 & {\center{\vspace{2mm}\textbf{Locality of gates needed for semi-universality\\ (w/o ancilla qudits)}}\\ } & {\center{\textbf{Number of ancilla qudits needed for  universality}}\\ }
   \\  
  \hline  
  \center{U(1)\\ }  &  \center{$\exp({\i \theta A })\ \  : \theta\in[0,2\pi)$\\  for a Hermitian operator  $A$ with equidistant integer eigenvalues} & \center{$k=2$}\\ 
 &  \ \ \ \ \ \ \ \ \ \ \ \ \ \  \ \ \ \ \ \ \  \vspace{3mm} 1
 \\  \hline
  \center{$\mathbb{Z}_p\ \ \ : p\ge 2$ }  &  $\exp(\frac{\i 2\pi a}{p} |1\rangle\langle 1|)\ \  :   a=0,\cdots, p-1$ & \center{$k=p$}&  \ \ \ \ \ \ \ \ \ \ \ \ \ \  \ \ \ \ \ \ \  1
 \\  \hline
    \center{Finite Abelian group $G$}\\  & \center{Regular representation}\\  & \center{$k=2$}\\  &  \ \ \ \ \ \ \ \ \ \ \ \ \ \  \ \ \ \ \ \ \  \vspace{3mm} 1
 \\  \hline
    \center{Finite Abelian group $G$}\\  & \center{ Arbitrary}\\  & \center{$2\le k \le |G|$}\\  &  \ \ \ \ \ \ \ \ \ \ \ \ \ \  \ \ \ \ \ \ \  \vspace{3mm} 1 \\  \hline
    \center{All Abelian groups}\\  & \center{ Arbitrary}\\  & \center{See 
    theorem \ref{Thm4}}\\  &  \ \ \ \ \ \ \ \ \ \ \ \ \ \  \ \ \ \ \ \ \  \vspace{3mm} 1 \\ 
  \hline
    \center{SU(2)} & \center{ Defining rep on $\mathbb{C}^2$} &$\ \ \ \ \ \ \ \ \ \ \ \  \ \ \ \   k=2$ &  \ \ \ \ \ \ \ \ \ \ \ \ \ \  \ \ \ \ \ \ \ 2
   \\ 
  \hline
   \center{SU($d$)\ \ \ \  : $d\ge 3$} & \center{Defining rep on $\mathbb{C}^d$} & $\ \ \ \ \ \ \ \ \ \  \ \ \ \  \ \  k=3$ &  \ \ \ \ \ \ \ \ \ \ \ \ \ \  \ \ \ \ \ \ $\le 3$  \\   
  \hline
\end{tabular}
   \caption{Conditions for semi-universality and universality for circuits with $k$-local symmetric unitaries.  The results on general Abelian symmetries, and, in particular, $\mathbb{Z}_p$ group and finite Abelian  groups, are established in this paper. The case of U(1) group was studied in \cite{marvian2022restrictions}, SU(2) group  in \cite{marvian2022rotationally} and SU($d$) for $d> 2$, was established in \cite{Marvian2021qudit} and \cite{hulse2024framework} (See also  \cite{marin2007algebre}).}
  \label{tab}
\end{table*}

Then, the next natural question is whether $k$-local symmetric unitaries are semi-universal? It turns out that the answer is not positive in general. That is, the locality of interactions impose a different type of constraints, namely 
\begin{enumerate}[label=\textbf{\Roman*}.]
  \setcounter{enumi}{1}
\item  There are observables commuting with all $k$-local symmetric unitaries that do not commute with some other symmetric unitaries. Equivalently, under the action of   $k$-local symmetric unitaries a charge sector $\mathcal{H}_\mu$ splits  into multiple irreducible  invariant subspaces, as 
\be\label{sub1}
\mathcal{H}_\mu=\bigoplus_\alpha \mathcal{H}_{\mu,\alpha}\ ,\ 
\ee
  such that all $k$-local symmetric unitaries are block-diagonal with respect to this decomposition.
\end{enumerate}

In particular, projectors to  the invariant subspaces $\{\mathcal{H}_{\mu,\alpha}\}$ are conserved under $k$-local symmetric unitaries, whereas they do not remain conserved under general symmetric unitaries. Hence, they are independent of the standard (Noether's) observables associated to the symmetry.  Interestingly, we find that in the case of Abelian symmetries, the only obstacle for semi-universality is the presence of these extra conserved observables. \\

\begin{figure}[t]
  \includegraphics[scale=.27]{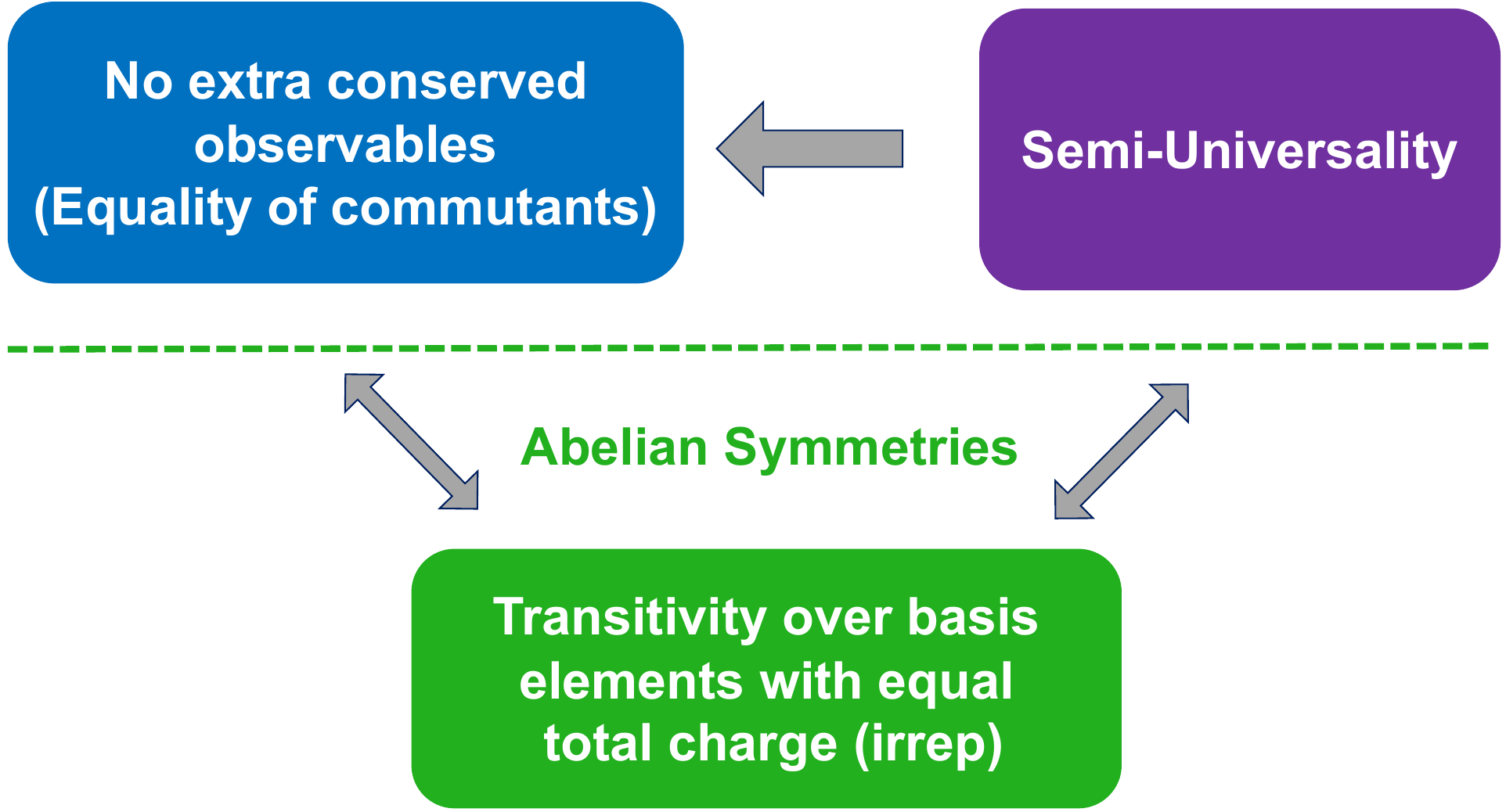}\vspace{2mm}
  \caption{\textbf{Relations  between three fundamental concepts:} We say $k$-local symmetric unitaries are semi-universal if they can realize all symmetric unitaries, up to possible constraints on the relative phases between sectors with different irreps (charges). This property implies that the only observables that commute with  $k$-local symmetric unitaries are those that commute with all symmetric unitaries. Our previous work in \cite{Marvian2021qudit}  shows that the reverse implication does not hold for general non-Abelian symmetries, such as SU($d$) with $d\ge 3$.  However,  as we show in this paper, the reverse implication \emph{does} hold  in the case of Abelian symmetries (See theorem \ref{Thm4} for a formal statement of this result.) Indeed, in this case the two aforementioned  properties are both equivalent to a third property, namely transitivity over basis elements with equal total irrep (charge), where the basis is chosen such that in each basis element each subsystem has a definite charge. Roughly speaking, this transitivity property means that $k$-local charge-conserving operations can redistribute the charge in the system.    }
  \label{FigRel}
\end{figure}

\noindent{\textbf{Theorem  (Informal version):} } For circuits with $k$-local $G$-invariant unitaries,     
if  $G$  is an Abelian group and $k\ge 2$, then following three properties  are equivalent:\\

\noindent$\bullet$ Semi-universality\\
\hspace{2mm}$\bullet$  No extra conserved observables (Equality of commutants)\\
\hspace{2mm}$\bullet$ Transitivity over basis elements with equal total charge (irrep). \\  

In the following we define and discuss the third property. See Theorem \ref{Thm4} for the formal statement and Fig. \ref{FigRel} for an illustration of this result.  

\subsection*{Transitivity:  The key  to understand Abelian circuits}

For Abelian symmetries each qudit has a basis in which the  action of the symmetry transformations are diagonal, i.e., each basis element has a definite charge (irrep). Furthermore, the $n$-fold tensor product of such states  defines a basis for the entire system, and since the group is Abelian, each tensor product state also has a definite charge, i.e., is  restricted to a single charge sector $\mathcal{H}_\mu$. For instance, in the case of U(1) symmetry corresponding to rotations around the z axis of a qubit, this basis is the so-called computational basis $\{|0\rangle,|1\rangle\}^{\otimes n}$. 

Two elements of this global  basis belong to the same charge sector $\mathcal{H}_\mu$ if, and only if, they can be converted to each other using a symmetric unitary, which can be interpreted as "charge redistribution" in the system.  
That is,  
symmetric unitaries act \emph{ transitively} on the basis elements in the same sector $\mathcal{H}_\mu$\footnote{In the above example of U(1) symmetry, the total charge (irrep) is determined by the Hamming weight of the bit string associated to each element of the computational basis, i.e., the number of 1’s in that element. }. On the other hand, in general, $k$-local symmetric unitaries with $k<n$ may not act transitively on such basis elements.  According to Theorem \ref{Thm4}, if they do act transitively then they are also semi-universal. In Sec. \ref{Sec:which} we go one step further and show that a subset of $k$-local symmetric Hamiltonians are semi-universal, if they can achieve transitivity (See corollary \ref{cor9}).

 These results significantly simplify  
 the study of semi-universality and universality in Abelian circuits; transitivity over basis elements  has a simple (classical) interpretation as the \emph{charge redistribution} in the system and  is generally a property that is  easier to check.\footnote{It is worth noting that the concept of transitivity can also be defined in classical logic circuits in a similar fashion. However, interestingly,  in the classical case transitivity does not imply universality. }   For instance, the semi-universality of 2-local U(1)-invariant unitaries, which was established in \cite{marvian2022restrictions} using a lengthy  Lie-algebraic argument specific to U(1) symmetry, can now be simply understood as a consequence of the fact that  2-qubit swap unitaries, which clearly respect the U(1) symmetry, act transitively on bit strings with equal Hamming weights (See Sec. \ref{Sec:Ex}).  

Remarkably, a very similar argument can now be applied for other Abelian groups. For example, in the case of  $\mathbb{Z}_2$ subgroup of this U(1) symmetry, a similar argument shows that 2-local unitaries that commute with $\sigma_z\otimes \sigma_z$ are semi-universal. More generally,  for any finite Abelian group $G$ with order $|G|\ge 2$, $k$-local symmetric unitaries with $k\ge |G|$ are semi-universal, i.e.,  generate all symmetric unitaries, up to type  
\textbf{I} constraints (See corollary \ref{cor3}). 
 It is also worth noting that, in general, depending on the symmetry and its representation,  universality may or may not be achievable with finite $k<n$. For instance, in  the example of qubit systems with the cyclic group $\mathbb{Z}_p$ symmetry with even $p$,  unless $k=n$, $k$-local  $\mathbb{Z}_p$-symmetric unitaries are not universal. In the case of odd $p$, on the other hand,   $k$-local $\mathbb{Z}_p$-symmetric unitaries  become universal for $k\ge p \ge 3$  (See table \ref{tab} and section \ref{Sec:Ex}). 



 
  
    Our results also allow us to characterize the realizable unitaries even when semi-universality, or, other equivalent properties does not hold.  Recall that when this condition is not satisfied some charge sectors $\mathcal{H}_\mu$ may further decompose into    
     invariant subspaces $\{\mathcal{H}_{\mu,\alpha}\}$, as in  Eq.(\ref{sub1}). Then, 
  theorem \ref{Thm2} implies that   
 any unitary $V$ that is block-diagonal with respect to the invariant subspaces $\{\mathcal{H}_{\mu,\alpha}\}$ can be realized with $k$-local symmetric unitaries with $k\ge 2$, provided that in all these  invariant subspaces,  the determinants of the realized unitaries are one. The latter condition guarantees that type \textbf{I} constraints are satisfied.

\subsection*{Elevating  semi-universality to universality using a single ancilla} 
 
We show that, in the case of Abelian symmetries using a single ancilla qudit one can always  elevate semi-universality to universality. Indeed, in section \ref{Sec:anc} we introduce a general mechanism for implementing unitary transformations that  are diagonal relative to the aforementioned basis using only 2-local symmetric unitaries and a single ancilla. This, in particular, 
 implies that one can circumvent type \textbf{I} constraints using a single ancilla qudit. 
 
 \subsection*{More restrictions for circuits with non-Abelian symmetries}
 
Remarkably, it turns out that in the case of non-Abelian symmetries,  the locality of interactions can impose other types of constraints, namely 
\begin{enumerate}[label=\textbf{\Roman*}.]
  \setcounter{enumi}{2}
  \item In certain  sectors $\{\mathcal{H}_\mu\}$ the realized unitaries are the orthogonal, the symplectic, or other (irreducible) subgroups of the symmetric unitaries.
\item The realized unitaries in certain sectors $\{\mathcal{H}_\mu\}$ dictate the unitaries in one or multiple  other sectors. In other words, in general, the time evolution of different sectors can not be  independent of each other. 
\end{enumerate}
Indeed, both of these restrictions exist in the case of SU($d$) symmetry with $d > 2$ \cite{Marvian2021qudit}. In particular, Ref. \cite{Marvian2021qudit} shows that even when the state of system is restricted to a sector $\mathcal{H}_\mu$ with an irreducible representation of SU($d$), certain functions of state, which are not conserved under general SU($d$)-invariant Hamiltonians, do remain conserved if the SU($d$)-invariant Hamiltonian can be decomposed as a sum of 2-local terms.  It is also worth noting that, in the special case of qubit systems with SU(2) symmetry, the above constraints do not exist, and the locality of interactions only imposes type \textbf{I} constraints \cite{marvian2022rotationally}. However, unlike the case of Abelian symmetries, in the case of SU(2) symmetry, a single ancilla qubit is not sufficient to elevate semi-universality to universality. Instead, universality can be achieved with 2 ancilla qubits  \cite{marvian2022rotationally}.\\


We proceed by first formally defining the framework of symmetric quantum circuits in Sec. \ref{sec:setup}, defining the no extra conserved observable condition in Sec. \ref{Sec:extra}, and presenting our main results in Sec. \ref{Sec:main}.   In Sec. \ref{Sec:which} we show  that if $k$-local symmetric unitaries are semi-universal, then only certain specific families of $k$-local symmetric gates are sufficient to achieve semi-universality.   In Sec. \ref{Sec:anc} we present a general mechanism for elevating semi-universality to universality using an ancilla qudit.  In Sec. \ref{Sec:discuss} we present further discussions  on the applications and the implications of symmetric quantum circuits with Abelian symmetries in the context of 
the resource theory of quantum thermodynamics \cite{FundLimitsNature, brandao2013resource, janzing2000thermodynamic,  lostaglio2015quantumPRX, halpern2016microcanonical, halpern2016beyond, guryanova2016thermodynamics, chitambar2019quantum}, thermalization of quantum systems with conserved charges, quantum reference frames \cite{QRF_BRS_07},  and 
universal quantum computing. The proofs of the main results are presented in section \ref{Sec:proof}.



\vspace{-3mm}

\color{red}



\color{black}

\section{Preliminaries }\label{sec:setup}

\subsection{The framework of symmetric quantum circuits}

We first briefly review the framework of symmetric quantum circuits (See \cite{marvian2022restrictions} for further details).  
Consider a system of $n$ qudits with the total Hilbert space $(\mathbb{C}^d)^{\otimes n}$.  We say an operator $A$ on this system is $k$-local, if, up to a permutation of qudits,  it  can be decomposed  as $A={A}_\text{loc}\otimes \mathbb{I}^{\otimes (n-k)}_d$, where ${A}_\text{loc}$ acts on $k$ qudits and $\mathbb{I}_d$ is the identity operator on a single qudit (See Fig.  \ref{Fig:circuit}).  Consider a group $G$ with a given unitary representation $u(g):g\in G$ 
on a single qudit. On $n$ qudits we consider the tensor product representation $U(g)=u(g)^{\otimes n}$. We say an operator $A$ is $G$-invariant, or, symmetric, if $[A, U(g)]=0$ for all $g\in G$.  Define $\mathcal{V}_{n,k}^{G}$  to be the set of all unitary transformations that can be implemented with $k$-local $G$-invariant  unitaries. More precisely, $\mathcal{V}_{n,k}^{G}$ is the set of unitaries $V=\prod_{i=1}^{m} V_{i}$,  generated  by composing $G$-invariant  $k$-local  unitaries $ V_i: i=1\cdots m$, for a finite $m$ \cite{marvian2022restrictions}. Equivalently, $\mathcal{V}_{n,k}^{G}$ is the set of unitary time evolutions $V(t)\ :t\ge 0$ of a system evolving under the  
Schr\"{o}dinger equation  
\be\nonumber
\frac{\d V(t)}{\d t}=-\i H(t) V(t)\ ,
\ee
 with the initial condition $V(0)=\mathbb{I}_d^{\otimes n}$,   for Hamiltonians $H(t)$ that are (i) $G$-invariant, and (ii) can be decomposed as a sum of $k$-local terms.  When the system has a given geometry,   one may impose the stronger requirement of geometric locality, i.e., $k$-local interactions should be restricted to $k$ nearest-neighbor qudits. However, if the qudits lie on a connected graph, such as a spin chain, this additional constraint does not change the group $\mathcal{V}_{n,k}^{G}$ \cite{marvian2022restrictions}.  It is well-known that in the absence of symmetries, i.e., when $G$ is the trivial group, 2-local unitaries are \emph{universal}  \cite{divincenzo1995two, lloyd1995almost}, i.e., $\mathcal{V}_{n,2}=\mathcal{V}_{n,n}$ for all $n$.

\begin{figure}[t]
  \includegraphics[scale=.4]{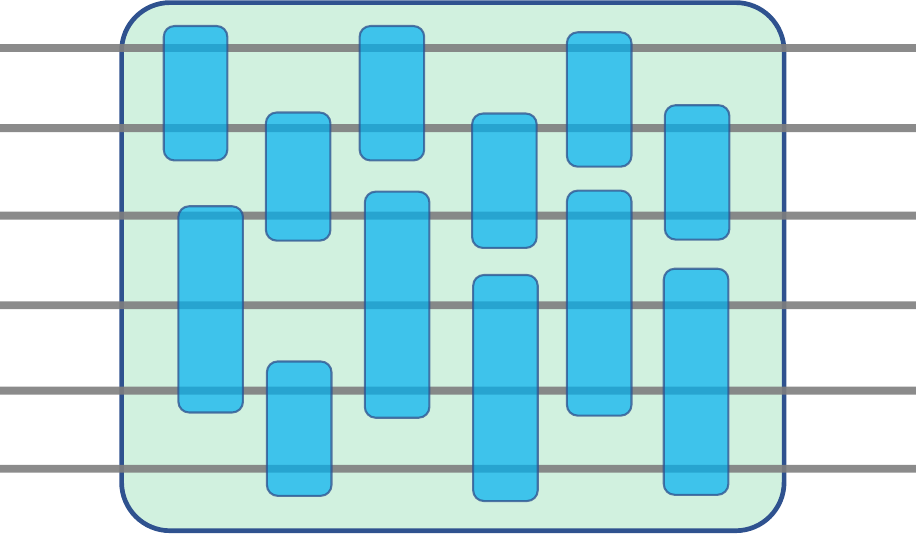}\vspace{2mm}
  \caption{\textbf{Example of a quantum circuit with 3-local gates:}  Each gate is a unitary transformation that acts on, at most, 3 qudits.   }
  \label{Fig:circuit}
\end{figure}

\subsection{A no-go theorem: Type \textbf{I} constraints}

 Does the universality of 2-local unitaries in the absence of symmetries  remain valid in the presence of symmetries? That is, does there exist a fixed $k$ such that $\mathcal{V}_{n,k}^G=\mathcal{V}_{n,n}^G$ for all $n$? The following result addresses this question. 

\begin{theorem} \label{Thm-1}\emph{\cite{marvian2022restrictions}} 
 For any integer $k\le n$, the group generated by $k$-local $G$-invariant unitaries on $n$ qudits, denoted by $\mathcal{V}_{n,k}^G, $ is a compact connected Lie group. The difference between the dimensions of this Lie group and the group of all $G$-invariant  unitaries is lower bounded by
\be\label{bound1}
\text{dim}(\mathcal{V}_{n,n}^G)-\text{dim}(\mathcal{V}_{n,k}^G)\ge |\text{Irreps}_G(n)|-|\text{Irreps}_G(k)|\ ,
\ee
where $|\text{Irreps}_G(k)|$ is the number of inequivalent irreps of group $G$  appearing in the representation $u(g)^{\otimes k}: g\in G$.
\end{theorem}
In the case of Abelian groups, the right-hand  side of Eq.(\ref{bound1}) can be interpreted as the difference between the total charge in the system and the charge that participates in $k$-local interactions.  For a Lie group $G$, the number of inequivalent irreps $|\text{Irreps}_G(n)|$ can grow unboundedly with the system size $n$, in which case the universality cannot be achieved with any finite $k$.  The bound in Eq.(\ref{bound1}) is a consequence of certain constraints on the relative phases between sectors with different charges (type \textbf{I} constraints), which are characterized in  \cite{marvian2022restrictions}  in terms of the Lie algebra associated to  $\mathcal{V}_{n,k}^G$ (Namely, they appear because  local symmetric interactions cannot generate certain elements of the center of the Lie algebra of all symmetric Hamiltonians).  In particular, these constraints imply that only for certain combinations of the phases $\{\exp(\i\phi_\mu)\}_\mu$ the unitaries  $\sum_\mu \exp(\i\phi_\mu)\Pi_\mu $ are in $\mathcal{V}_{n,k}^G$, where $\Pi_\mu$ is the projector to the sector $\mathcal{H}_\mu$ with charge $\mu$. In Appendix \ref{Sec:App1} we present further discussions and a summary of the derivation of Eq.(\ref{bound1}).     

\subsection{Definition: Semi-universality}\label{def:semi}

In addition to type \textbf{I} constraints highlighted in the above theorem,  which appear for general symmetries, for some symmetries the locality of interactions imposes further restrictions on the realizable unitaries. To focus  on these additional constraints and ignore the  type \textbf{I} constraints, in the following  
we define the notion of
  \emph{semi-universality}, which is weaker than universality: 
  For a general group $G$ (Abelian, or non-Abelian)  $k$-local $G$-invariant unitaries are called \emph{semi-universal}  if 
 for any $G$-invariant unitary $V\in \mathcal{V}^G_{n,n}$, there exists a set of phases $\phi_{\mu}\in(-\pi,\pi]$, such that the unitary  $ (\sum_\mu e^{\i\phi_\mu} \Pi_\mu)V\in\mathcal{V}^G_{n,k}$, which means it is realizable with $k$-local $G$-invariant unitaries. In  words, this condition means that   $k$-local $G$-invariant unitaries generate  all $G$-invariant unitaries, up to constraints on the relative phases between sectors with different charges.   This condition can be equivalently stated as\footnote{Note that the Lie algebra associated to $\mathcal{SV}^G_{n,k}\equiv \big[\mathcal{V}^G_{n,k}\ ,\mathcal{V}^G_{n,k}\big]$
 is the semi-simple part of the Lie algebra 
 associated to $\mathcal{V}^G_{n,k}$.} 
 \be\label{semi5}
 \mathcal{SV}^G_{n,k}\equiv \big[\mathcal{V}^G_{n,k}\ ,\mathcal{V}^G_{n,k}\big]=\big[\mathcal{V}^G_{n,n}\ ,\mathcal{V}^G_{n,n}\big]\ ,
 \ee
 where  $[\mathcal{V}^G_{n,k},\mathcal{V}^G_{n,k}]$ denotes the commutator subgroup of group $\mathcal{V}^G_{n,k}$, i.e., the subgroup generated by the unitaries  in the form $V_1V_2V_1^{\dag}V_2^{\dag}$ for arbitrary $V_1, V_2\in \mathcal{V}^G_{n,k}$.   It is worth noting that, from a Lie-algebraic perspective, Eq. (\ref{semi5}) implies that the semi-simple parts of the Lie algebras associated with $\mathcal{V}^G_{n,k}$ and $\mathcal{V}^G_{n,n}$ are identical. However, in general, the centers of these Lie algebras may differ, which will result in constraints of type $\textbf{I}$ (See Appendix \ref{Sec:App1} and Ref. \cite{marvian2022restrictions} for further discussions).

  In the special case of Abelian groups,  which is the main focus of this paper, we have      
 \be\label{Ab-inc}
 \mathcal{V}^G_{n, k}\subseteq  \mathcal{V}^G_{n,n}=\bigoplus_{\mu\in\text{Irreps}_G(n)} \hspace{-3mm}\text{U}(\mathcal{H}_{\mu})\ ,
 \ee 
  where the summation is over inequivalent irreps of group $G$ appearing in the representation of the symmetry on $n$ qudits, and 
$\text{U}(\mathcal{H}_{\mu})$ is the group of all unitaries on space $\mathcal{H}_{\mu}$.  In this case, semi-universality in Eq.(\ref{semi5}) can be equivalently rephrased as
 \be\label{semi}
  \mathcal{SV}^G_{n,n} =\bigoplus_{\mu\in\text{Irreps}_G(n)} \hspace{-3mm}\text{SU}(\mathcal{H}_{\mu})  \subset \mathcal{V}^G_{n,k}\ ,
\ee
where $\text{SU}(\mathcal{H}_{\mu}) $ is the group of special unitaries on $\mathcal{H}_{\mu}$, i.e., unitaries with determinant 1. 


\section{Extra conserved Observables}\label{Sec:extra}

Besides the constraint on the relative phases, there can be another obstruction to the universality of local symmetric unitaries, namely the presence of extra conserved observables.

\vspace{-2mm}\subsection{Example:   Qubit systems with cyclic symmetry}\label{Sec:cyc}
For any integer $p$,  the single-qubit unitary transformations 
\be\nonumber
u(a)=\exp(\frac{\i 2\pi a}{p}  |1\rangle\langle 1|)=\left(
\begin{array}{cc}
1  &     \\ & e^{\frac{\i 2\pi a}{p} }
  \end{array}
\right)\ ,
\ee
define a representation of the cyclic group $\mathbb{Z}_p$, corresponding to integers $a=0,1,\cdots, p-1$ with addition mod $p$.    The action of this symmetry  on $n$ qubits is given   by unitaries ${u}(a)^{\otimes n}$. For $n<p$  the set of $\mathbb{Z}_p$-invariant unitaries  coincides with the set of unitaries  respecting the stronger U(1) symmetry, corresponding to rotations around the z axis, which is represented by  unitaries $(\exp(\i\theta \sigma_z/2))^{\otimes n}$ for  $\theta\in[0,2\pi)$ (See section \ref{Sec:Ex}).  On the other hand, for $n\ge p$ there are unitaries that respect the $\mathbb{Z}_p$ symmetry, but not this U(1) symmetry. This is a consequence of the fact that both states $|1\rangle^{\otimes p}$ and $|0\rangle^{\otimes p}$  remain invariant under the action of $\mathbb{Z}_p$  group 
and, therefore, can be interconverted into each other by  unitaries that respect this symmetry  
(in other words, $p$ copies of "excitation" $|1\rangle$ annihilate each other into the  "vacuum"). It follows that for $n\ge p$ the observable $O=|1\rangle\langle 1|^{\otimes n}$ does not remain conserved under general $\mathbb{Z}_p$-invariant unitaries, whereas it does remain conserved under $k$-local $\mathbb{Z}_p$-invariant unitaries for $k<p$. In this example, it is clear that to achieve semi-universality,  $k$-local $\mathbb{Z}_p$-invariant  unitaries with $k\ge p$ are needed. 

\subsection{No extra conserved observables condition\\ (The equality of commutants)}


The above example clearly shows that in the presence of a symmetry the locality of interactions can impose additional conservation laws,  independent of the standard conservation laws  associated to the symmetry. In the following, we say there are  no extra conserved observables 
under $k$-local symmetric  unitaries if 
\be\label{Noeth}
\text{Comm}\{\mathcal{V}^G_{n, k}\}\stackrel{?}{=}\text{Comm}\{\mathcal{V}^G_{n, n}\}=\text{Span}_\mathbb{C}\{U(g): g\in G\}\ ,
\ee 
where  the left-hand side is the space of operators commuting with all unitaries  in $\mathcal{V}^G_{n, k}$, and the second equality always holds by the bicommutant theorem. 
Observables that belong to the linear space in the right-hand side  are  the standard (Noether's) conserved observables associated to the symmetry. That is, they are conserved under all symmetric unitaries (It is  also worth mentioning that  any unitary that commutes with all such observables is a symmetric unitary, i.e., belongs to $\mathcal{V}^G_{n, n}$).




As we saw before in Eq.(\ref{Ab-inc}),  for an Abelian group $G$, symmetric unitaries are those that are block-diagonal with respect to subspaces with different charges. Then, the right-hand side of Eq.(\ref{Noeth}) is equal to the span of projectors to $\{\mathcal{H}_{\mu}\}$,  and therefore has dimension $|\text{Irreps}_G(n)|$. Hence, the  condition in Eq.(\ref{Noeth})  can be equivalently stated as
\be\label{Abs}
\text{dim}(\text{Comm}\{\mathcal{V}^G_{n, k}\})\stackrel{?}{=}|\text{Irreps}_G(n)|\ .
\ee
%
%

\section{(Non-)Universality in Abelian circuits}\label{Sec:main}
According to our first main result, in the case of Abelian symmetries  if there are no extra conserved observables  then all symmetric unitaries are realizable, up to constraints on the relative phases (type \textbf{I} constraints).

\begin{theorem}\label{Thm3}
For an Abelian group $G$, the no-extra-conserved-observable condition  in Eq.(\ref{Noeth}) (or, equivalently in Eq.(\ref{Abs})) holds for $k\ge 2$ if, and only if, 
\be
\bigoplus_{\mu\in\text{Irreps}_G(n)} \hspace{-3mm}\text{SU}(\mathcal{H}_{\mu}) \subset \mathcal{V}^G_{n,k} \ .
\ee
\end{theorem}
This equation means that $k$-local $G$-invariant unitaries are semi-universal.  Note that the celebrated universality of 2-local unitaries   
in the absence of symmetries \cite{divincenzo1995two, lloyd1995almost}, corresponds to a special case of this theorem,  when the representation of group $G$ on the system is trivial. 

 The proof of this theorem is presented in section \ref{Sec:proof}. To establish this result, first we show the following more general theorem, which does not rely on the no extra conserved observable condition. 

\begin{theorem}\label{Thm2}
For an Abelian group $G$, under the action of $k$-local $G$-invariant unitaries  with $k\ge 2$,   the total Hilbert space of $n$ qudits decomposes into  orthogonal subspaces $\{\mathcal{H}_{\mu,\alpha}\}$ as 
\be\label{decomp}
(\mathbb{C}^d)^{\otimes n}\cong \bigoplus_{\mu\in\text{Irreps}_G(n)} \mathcal{H}_\mu=\bigoplus_{\mu\in\text{Irreps}_G(n)} \bigoplus_\alpha \mathcal{H}_{\mu,\alpha}\ ,
\ee
such that $\mathcal{V}^G_{n,k}$ is block-diagonal with respect to this decomposition and  \be\label{subgroup}
\bigoplus_{\mu,\alpha} \text{SU}\left(\mathcal{H}_{\mu,\alpha}\right)\ \subset \mathcal{V}^G_{n,k} \subseteq \ \bigoplus_{\mu,\alpha} \text{U}\left(\mathcal{H}_{\mu,\alpha}\right)\ .
\ee 
\end{theorem}
Equivalently, Eq.(\ref{subgroup}) can be restated as
\be\label{def:commut}
\mathcal{SV}^G_{n,k}\equiv \big[\mathcal{V}^G_{n,k}\ ,\mathcal{V}^G_{n,k}\big]=\bigoplus_{\mu,\alpha} \text{SU}\left(\mathcal{H}_{\mu,\alpha}\right)\ ,
\ee 
which fully characterizes the commutator subgroup of $\mathcal{V}^G_{n,k}$, denoted by $\mathcal{SV}^G_{n,k}$. 

In summary, in the case of Abelian symmetries the locality of interactions only imposes constraints of types \textbf{I} and  \textbf{II} (i.e., no  types \textbf{III} and  \textbf{IV} constraints). Furthermore, the no extra conserved observable condition rules out type \textbf{II} constraints, which means  the only constraints on the realizable unitaries are of type \textbf{I}.



\subsection{Transitivity implies Semi-universality.}

In the example of the cyclic symmetry we saw that the extra conserved observables are related to the fact that, due to the locality of interactions, the charge
associated to the symmetry cannot be arbitrarily redistributed in the system. In the following,  
we show that this interpretation can be generalized to all Abelian symmetries.   

An important property of  Abelian symmetries, which plays a crucial role in all the arguments in this paper, is the additivity of the charge associated to the symmetry. In particular, the charges (irreps) of the subsystems uniquely determine the total charge in the system. Let $|r\rangle: r=0,\cdots, d-1$ be an orthonormal basis for $\mathbb{C}^d$, 
with the property that  single-qudit unitaries $\{u(g): g\in G\}$ are simultaneously diagonal in this basis (Note that such basis exists because the group is Abelian).  The $n$-fold tensor product of these states define an orthonormal basis for the total Hilbert space $(\mathbb{C}^d)^{\otimes n}$, denoted as 
\be\label{form}
\pmb{B}_n=\big\{{|\pmb{r}}\rangle=|r_1\rangle\otimes \cdots \otimes |r_n\rangle\ \  : r_j=0,\cdots, d-1\big\}\ ,
\ee
where $\pmb{r}=r_1\cdots r_n$. 
Each basis element is an eigenvector of unitaries $U(g)=u(g)^{\otimes n}$, and therefore is a vector in a single charge sector $\mathcal{H}_\mu$.   In the following, the Hamming distance between  $\pmb{r}=r_1\cdots r_n$, and  $\pmb{r}'=r'_1\cdots r'_n$, denoted by $d(\pmb{r}, \pmb{r}')$, is  the  
 number of qudits which are assigned different reduced states by $|\pmb{r}\rangle$ and $|\pmb{r}'\rangle$. The following result extends Theorem \ref{Thm3} by including an equivalent condition.\\

 \begin{thmbis}{Thm3}\label{Thm4} 
For an Abelian group $G$, let $\bigoplus_{\mu} \mathcal{H}_\mu$ be the decomposition of the total Hilbert space of $n$ qudits into subspaces with inequivalent irreps (charges) of group $G$. For $k\ge 2$, the  following statements are equivalent:
\begin{enumerate}
\item  \textbf{Semi-universality:} 
$${\bigoplus_{\mu\in\text{Irreps}_G(n)} \text{SU}(\mathcal{H}_{\mu})\subset \mathcal{V}^G_{n,k}}\ . $$\vspace{-4mm}
\item \textbf{No extra conserved observable condition:} Eq.(\ref{Noeth}), or, equivalently, Eq.(\ref{Abs}) holds.
\item \textbf{Transitivity over the basis elements with equal  charge (irrep):} For any pair of basis elements  $|\pmb{r}\rangle, |\pmb{r}'\rangle\in \pmb{B}_n$ that belong to the same charge sector $\mathcal{H}_\mu$, there exists a sequence of  elements of  $\pmb{B}_n$ connecting 
$|\pmb{r}\rangle$ to $|\pmb{r}'\rangle$  as 
\be\nonumber
|\pmb{r}\rangle=|\pmb{s}^1\rangle \longrightarrow  |\pmb{s}^2\rangle \longrightarrow  \cdots \cdots\longrightarrow  |\pmb{s}^t\rangle=|\pmb{r}'\rangle\ ,
\ee
such that  (i) all states $|\pmb{s}^j\rangle$ are in the same charge sector $\mathcal{H}_{\mu}$, and (ii)  the Hamming distance between any consecutive pair $\pmb{s}^{j}$ and $\pmb{s}^{j+1}$ is 
 $d(\pmb{s}^{j}, \pmb{s}^{j+1})\le k$. 
\end{enumerate}
\end{thmbis}


The condition in statement 3 has a simple (classical) interpretation: it means that the charge associated to the symmetry can be arbitrarily redistributed in the system via a sequence of  $k$-local charge-conserving   operations. According to the theorem, if this property holds, 
then there are no extra conserved observables and  all the symmetric unitaries  are realizable by $k$-local symmetric unitaries, up to relative phases between sectors with different charges (type \textbf{I} constraints).   


It is also worth noting that 
 using swap unitaries,  which are 2-local and $G$-invariant, any state $|\pmb{r}\rangle\in \pmb{B}_n$ can be mapped to an arbitrary permuted version of this state. Therefore, for testing the condition in statement 3, the only relevant property of  $|\pmb{r}\rangle$ and $|\pmb{r}'\rangle$ is the number of qudits 
in each irrep $\mu\in\text{Irreps}_G(1)$ (That is, the order of qudits does not matter).   
  In the rest of this section, we discuss the implications of these theorems.

 
%
%
%
%

\subsection{Finite Abelian groups}
In the case of finite  groups, $|\text{Irreps}_G(n)|$ is bounded by the order of the group, denoted by $|G|$. 
Let $l_\text{min}$ be the smallest positive integer satisfying
\be\label{lmin}
|\text{Irreps}_G(l_\text{min})|=|\text{Irreps}_G(n)|\le |G|\ .
\ee
Roughly speaking, this means that the total charge in the system with $n$ qudits can be compressed into $l_\text{min}$ qudits. Then, it can be easily seen that statement 3 of theorem \ref{Thm4} holds  for $k=l_\text{min}+1$:  one can use $l_\text{min}$ qudits as a \emph{charge reservoir} and by coupling them sequentially to all other qudits in the system via $(l_\text{min}+1)$-local $G$-invariant  unitaries,  
 one can transform any basis element  $|\pmb{r}\rangle$ to $|\pmb{r}'\rangle$, provided that they have equal total charges.   We conclude that the three statements in theorem \ref{Thm4} hold for $k\ge l_\text{min}+1$. For instance, if each qudit carries the regular representation of the group, then $l_\text{min}=1$, independent of $n$, which means 2-local $G$-invariant unitaries are semi-universal.
 
 Next, we note that  for all $l<l_\text{min}$, $|\text{Irreps}_G(l)|$  monotonically increases with $l$. More precisely,    $|\text{Irreps}_G(l+1)|>|\text{Irreps}_G(l)|$ (See Appendix \ref{Sec:App1}). Furthermore, assuming each qudit has more than a single charge sector, which means for some group element $g\in G$ the unitary $u(g)$ is not a global phase, then $|\text{Irreps}_G(1)|\ge 2$. Together with Eq.(\ref{lmin}) this means
\be
l_\text{min}\le |G|-1\ .
\ee 
Then, applying theorem \ref{Thm4}, we arrive at 
\begin{corollary}\label{cor3}
Let $G$ be an arbitrary finite Abelian group with order $|G|$. For $k\ge \text{min}(|G|, 2)$, $k$-local $G$-invariant unitaries are semi-universal, i.e., Eq.(\ref{semi}) holds.
\end{corollary}

\subsection{Type \textbf{I} constraints}\label{sec:typeI}

Combining theorem \ref{Thm4} with
the results of \cite{marvian2022restrictions} discussed in Appendix \ref{Sec:App1}, it can be shown that for a connected Abelian group $G$, such as U(1), 
 if  any/all of the equivalent conditions in theorem \ref{Thm4} are  satisfied, then the bound in Eq.(\ref{bound1}) holds as equality, i.e.,
\be\label{bound2}
\text{dim}(\mathcal{V}_{n,n}^G)-\text{dim}(\mathcal{V}^G_{n,k})= |\text{Irreps}_G(n)|-|\text{Irreps}_G(k)|\ .
\ee
Furthermore, this equation  holds even if $G$ is not 
 connected, but $\Tr(u(g))\neq 0$ for all $g\in G$  (See Appendix \ref{Sec:App1}).

It is also worth noting that when $G$ is a connected Abelian group, if each qudit  has more than a single charge sector, such that for some group elements $g\in G$, $u(g)$ is not a global phase, then, $|\text{Irreps}_G(k)|< |\text{Irreps}_G(n)|$ for $k<n$ (See Appendix \ref{Sec:App1}). This implies that, unless $k=n$, general symmetric unitaries cannot be realized with $k$-local symmetric unitaries\footnote{On the other hand, in the case of non-Abelian groups such as SU(2), it is possible to achieve $|\text{Irreps}_G(k)|= |\text{Irreps}_G(n)|$ for $k<n$  \cite{marvian2022rotationally}.}.  The latter statement also holds  even if $G$ is not connected but  $\Tr(u(g))=0$ for a group element $g\in G$ 
 (See Appendix \ref{Sec:App1}).

\subsection{Examples}\label{Sec:Ex}

\noindent\textbf{U(1) group:} For a system of $n$ qubits, consider the U(1) symmetry corresponding to rotations around the z axis, $\exp(\i\theta \sigma_z/2): \theta\in[0,2\pi)$. Then, the basis $\pmb{B}_n=\{|0\rangle,|1\rangle\}^{\otimes n}$ is the standard (computational) basis for $n$ qubits. Any pair of bit strings with equal  Hamming weights (i.e., the same number of 1's) can be converted to each other by a sequence of  swaps, which are 2-local  and respect the U(1) symmetry. Then, theorem \ref{Thm4} immediately implies that 2-local U(1)-invariant unitaries are semi-universal. Furthermore, 
because U(1) is a connected group,  Eq.(\ref{bound2})
implies that for $k\ge 2$,
\be
\text{dim}(\mathcal{V}_{n,n}^{U(1)})-\text{dim}(\mathcal{V}^{U(1)}_{n,k})=n-k \ .
\ee
These results were previously obtained in \cite{marvian2022restrictions}, using Lie-algebraic arguments that are specific to U(1) symmetry on qubits (namely, by considering the nested commutators of the Hamiltonian  $\sigma_x\otimes \sigma_x+\sigma_y\otimes \sigma_y$). \\

\noindent\textbf{Cyclic groups:}  Recall the example of qubit systems with the cyclic group  $\mathbb{Z}_p$ in section \ref{Sec:cyc}. This group has $p$ distinct 1D irreps 
\be\nonumber
f_s(a)=\exp(\frac{\i 2\pi a s}{p})\ \ \ \  : s=0,\cdots, p-1\ ,
\ee
for arbitrary group element $a\in\{0,\cdots, p-1\}$.   The total $\mathbb{Z}_p$ charge of the $n$-qubit state $|\pmb{r}\rangle=|r_1\rangle\cdots |r_n\rangle\in\{|0\rangle,|1\rangle\}^{\otimes n}$ is determined by
\be\nonumber
s_{\text{tot}}=\sum_{j=1}^n r_{j}\ \ \ \ \ \ (\text{mod} \ p)\ .
\ee 
For a system with $n$ qubits, if $n < p$ then the value of this charge  uniquely determines the Hamming weight 
of $r_1\cdots r_n$. This means that under $\mathbb{Z}_p$-invariant gates acting on $k<p$ qubits,  the Hamming weight $\sum_{j=1}^n r_{j}$ is conserved.  On the other hand,  since $|0\rangle^{\otimes p}$ and $|1\rangle^{\otimes p}$ have equal total $\mathbb{Z}_p$ charges,   using $k$-local $\mathbb{Z}_p$-invariant  unitaries with $k\ge p$, one can transform  any basis element 
 $|\pmb{r}\rangle$ to a "standard"  form with Hamming weight  in the range $0,\cdots, p-1$.
 It follows that condition 3 in theorem \ref{Thm4} is satisfied for $k\ge p$. Therefore, $k$-local $\mathbb{Z}_p$-invariant  unitaries are semi-universal, which can also be seen using corollary \ref{cor3}.
  
   Interestingly, in this case type \textbf{I} constraints depend on whether 
$p$ is odd or even. For even $p$, we have $\Tr(u(p/2))=0$, whereas for odd $p$, $\Tr(u(a))\neq 0$ for all $a\in\{0,\cdots, p-1\}$. Then, the results of section \ref{sec:typeI} imply 
\begin{align}
p \text{  is odd } \ \ \ &\Longrightarrow\ \ \  \mathcal{V}_{n,k}=\mathcal{V}_{n,n}\ \ \  \ \  \  : k\ge p \nonumber \\
\ p \text{ is even} \ \ \ &\Longrightarrow\ \ \  \mathcal{V}_{n,k}\neq \mathcal{V}_{n,n}\nonumber\ ,  \ \  \  \  : k< n\ ,
\end{align}
where we have omitted the superscript $\mathbb{Z}_p$ (We note that independent of and after the present work, Ref. \cite{jain2023locality} 
discusses this example and reports similar results).

\section{Which $k$-local $G$-invariant gates are needed?} \label{Sec:which}

So far, our focus has been on the group generated by all $k$-local $G$-invariant unitaries, denoted as $\mathcal{V}^G_{n,k}$. However, upon closer examination of the constructive proof of Theorem \ref{Thm2} presented in Sec.  \ref{Sec:proof} (which also implies Theorems \ref{Thm3} and \ref{Thm4}), it becomes clear that only specific families of $k$-local $G$-invariant unitaries are required. These restricted families may not necessarily generate all elements of $\mathcal{V}^G_{n,k}$. But, they still  generate its commutator subgroup $\mathcal{SV}^G_{n,k}$ in Eq.(\ref{def:commut}). The following corollary presents a list of all $k$-local $G$-invariant gates that are needed in this construction, to generate $\mathcal{SV}^G_{n,k}$ (Note that since $\mathcal{V}^G_{n,k}$ itself is contained in $\bigoplus_{\mu,\alpha} \text{U}\left(\mathcal{H}_{\mu,\alpha}\right)$,  the extra $k$-local $G$-invariant gates that are not  included in the following list can only affect the relative phases between sectors $\{\mathcal{H}_{\mu,\alpha}\}$).  


In the following, for any pair of basis elements $|\textbf{r}\rangle, |\textbf{r}'\rangle\in \pmb{B}_k$, we say $|\textbf{r}\rangle$ and $|\textbf{r}'\rangle$ have the same total charge if
\be\label{equal-charge}
 \langle\pmb{r}|u(g)^{\otimes k}|\pmb{r}\rangle= \langle\pmb{r}'|u(g)^{\otimes k}|\pmb{r}'\rangle\ \ : \forall g\in G\ ,
 \ee
which means they carry the same irrep of group $G$.

\begin{corollary}\label{cor9}
Any unitary in the group $\mathcal{SV}^G_{n,k}$ can be realized using the following 3 families of $G$-invariant unitaries:  
\begin{enumerate}
\item Single-qudit unitaries  $\exp(\frac{\i\pi }{2}|r\rangle\langle r|)$, for all $|r\rangle\in \pmb{B}_1$.
\item 2-qudit  Hermitian unitaries $C_{12}=\mathbb{I}_d^{\otimes 2}-2|r_{1}\rangle\langle r_{1}|_{1}\otimes |r_2\rangle\langle  r_2|_2$, for all $|r_{1}\rangle, |r_{2}\rangle\in \pmb{B}_1$.

\item The family of unitaries  $\exp(\i \theta \mathbb{X}(\textbf{t} ; \textbf{t}'))$ for $\theta\in[0,2\pi)$ and $\mathbb{X}(\textbf{t} ; \textbf{t}')\in \textbf{H}_\text{redist}$, where $\textbf{H}_\text{redist}$ is any set  of $G$-invariant Hermitian operators in the form $\mathbb{X}(\textbf{t} ; \textbf{t}')=|\textbf{t}\rangle\langle \textbf{t}'|+|\textbf{t}'\rangle\langle \textbf{t}|$ with  $|\textbf{t}\rangle, |\textbf{t}'\rangle\in \pmb{B}_l$ for some $l\le k$ satisfying the following property:    $\textbf{H}_\text{redist}$ acts transitively on the elements of $\pmb{B}_k$ that have the same total charge, i.e.,   for any  $|\textbf{r}\rangle, |\textbf{r}'\rangle\in \pmb{B}_k$, if Eq.(\ref{equal-charge}) holds, then there exists a sequence of elements   
$\mathbb{X}(\textbf{t}^{w+1} ; \textbf{t}^w),\cdots,  \mathbb{X}(\textbf{t}^{2} ; \textbf{t}^{1})  \in  \textbf{H}_\text{redist}$ such that
\be\label{trans4}
|\pmb{r}'\rangle=\mathbb{X}_{A_w}(\textbf{t}^{w+1} ; \textbf{t}^w)\cdots  \mathbb{X}_{A_1}(\textbf{t}^2 ; \textbf{t}^1)|\pmb{r}\rangle\ ,
\ee 
where $\mathbb{X}_{A_j}(\textbf{t}^{j+1} ; \textbf{t}^j)$ acts as $\mathbb{X}(\textbf{t}^{j+1} ; \textbf{t}^j)$ on the subset  $A_j$ of qudits and acts trivially on the rest of qudits.

\end{enumerate}
\end{corollary}

Note that the first two families of unitaries are independent of parameter $k$ and the structure of group $G$. For instance, for the above examples of U(1) and its cyclic subgroup $\mathbb{Z}_p$, these families are identical. Furthermore, the property that is required for the third family is that  $\textbf{H}_\text{redist}$ should act transitively  on the elements of $\pmb{B}_k$ that have the same charge. In general, for $k<n$ , $\textbf{H}_\text{redist}$ does not act transitively on elements  of  $\pmb{B}_n$ that belong to the same charge sector.  If it does, then corollary \ref{cor9} implies that $\mathcal{SV}^G_{n,k}=\mathcal{SV}^G_{n,n}$, which means semi-universality can be achieved with the above three gate sets (This, in particular, implies  theorem  
\ref{Thm3}). 


It is also worth noting  that under the first two families of gates the charge of each qudit is conserved. That is, such unitaries commute with $\bigotimes_{i=1}^n u(g_i)$ for all $g_i\in G$. On the other hand, the third family of unitaries, which correspond to Hamiltonians in $\textbf{H}_\text{redist}$, redistribute the charge in the system, which is crucial  for semi-universality.



\subsection*{Examples}

\noindent\textbf{No symmetries:}  
A special case of this framework is a system of qubits  with no symmetries, or, equivalently, with the trivial representation of the symmetry. In this case we can choose $ \pmb{B}_n=\{|0\rangle,|1\rangle\}^{\otimes n}$, which is the computational basis. Then, up to  a global phase,   the first family of unitaries in the above gate sets is generated by the single-qubit gate $S=\exp(\i\frac{\pi}{2}|1\rangle\langle 1|)$ (Note that $\exp(\i\frac{\pi}{2}|0\rangle\langle 0|)=\i S^3$). The second family is the 2-qubit Controlled-Z gate, also denoted as  $\text{CZ}_{12}=\mathbb{I}^{\otimes 2}-2|1\rangle\langle 1|_{1}\otimes |1\rangle\langle  1|_2$, and 
three other  2-qubit gates obtained from this gate by sandwiching it between Pauli x operators, namely 
\be\label{CZ4}
(\sigma^{r_1}_x\otimes \sigma^{r_2}_x) \text{CZ}_{12}(\sigma^{r_1}_x\otimes \sigma^{r_2}_x)\ \ \ \ \  : \ r_1,r_2=0,1\ . 
\ee  
The third family can be chosen to be the single-qubit gates $\exp(\i \theta \sigma_x): \theta\in[0,2\pi)$, which corresponds to $\textbf{H}_\text{redist}=\{\mathbb{X}(0;1)=\sigma_x\}$. In this context the transitivity condition in Eq.(\ref{trans4}) is satisfied trivially for $k=n$ because all basis elements    
have equal charge, i.e., Eq.(\ref{equal-charge}) holds trivially, and one can go from any element of the computational basis to another via a sequence of Pauli x operators. Using single-qubit gates $\exp(\i \frac{\pi}{2} \sigma_x)=\i\sigma_x$ one also obtains all unitaries in Eq.(\ref{CZ4}) from CZ gate. 

In summary, in this special case our result implies that  $S$ and CZ gates together with single-qubit rotations around the x axis are universal, i.e., generate all $n$-qubit unitaries, up to a global phase (This, of course, also follows from the well-known universality of CNOT and single-qubit gates). This example also demonstrates that, for general groups and representations, semi-universality cannot be achieved with any combination of two among the three families of gates listed in corollary \ref{cor9}. \\


\noindent\textbf{U(1) group:}  In the example of  qubit systems with U(1) symmetry discussed in Sec. \ref{Sec:Ex},  $\textbf{H}_\text{redist}=\{\mathbb{X}(01;10)\}$  with
\be
\mathbb{X}(01;10)=|01\rangle\langle 10|+|10\rangle\langle 01|=\frac{1}{2}(\sigma_x\otimes \sigma_x+\sigma_y\otimes \sigma_y)\ ,
\ee
satisfies the transitivity condition in Eq.(\ref{trans4}) for all $k\le n$. That is, by applying a sequence of this 2-qubit operator on different pairs of qubits we can go from any  element of the computational basis $\pmb{B}_n$ to any other element with equal Hamming weight. Then, corollary \ref{cor9} implies that  the semi-universality can be achieved using $S$ gate, 2-qubit gates in Eq.(\ref{CZ4}), and the family of unitaries $\exp(\i \theta (\sigma_x\otimes \sigma_x+\sigma_y\otimes \sigma_y)): \theta\in[0,2\pi)$.    \\

\noindent\textbf{Cyclic group:}  The above three gate sets that are semi-universal for 
U(1) symmetry are not sufficient to achieve semi-universality in the case of  
$\mathbb{Z}_p$ symmetry discussed in 
section \ref{Sec:cyc}. In particular, with such gates the transitivity condition in Eq.(\ref{trans4}) is not satisfied. To achieve transitivity, and hence semi-universality, we amend these gate sets with the gates $\exp(\i\theta \mathbb{X}(0^p;1^p)): \theta\in[0,2\pi)$, where 
\be\nonumber
\mathbb{X}(0^p;1^p)=|0\rangle\langle 1|^{\otimes p}+|1\rangle\langle 0|^{\otimes p}\ .
\ee
It can be easily seen that $\mathbb{X}(0^p;1^p)$ together with $\mathbb{X}(01;10)$ act transitively on the elements of $\pmb{B}_n$ with equal charge: Recall that for $\mathbb{Z}_p$ symmetry  the charge (irrep) is uniquely determined by the Hamming weight mode $p$. Using a sequence of operators $\mathbb{X}(01;10)$ it is possible to convert any basis element in  $\pmb{B}_n$  to another basis element with equal Hamming weight. Furthermore,  applying $\mathbb{X}(0^p;1^p)$  on $p$ qubits  which are all in state $|0\rangle$ (or, all in state $|1\rangle$) increases (or, decreases) the Hamming weight by $p$. Combining these operations together one can achieve transitivity on the basis elements with equal Hamming weights mode $p$. 

For instance, in the special case of bit-parity $\mathbb{Z}_2$ symmetry, we have
\be
\mathbb{X}(00;11)=\frac{1}{2}(\sigma_x\otimes  \sigma_x-\sigma_y\otimes  \sigma_y)\ .
\ee
Furthermore, since $S \sigma_x S ^\dag=\sigma_y$, using the $S$ gates together with the family of unitaries $\exp(\i \theta \sigma_x\otimes \sigma_x):\theta\in [0,2\pi)$ we obtain both families of unitaries  $\exp(\i \theta \mathbb{X}(01,10))$ and $\exp(\i \theta \mathbb{X}(00,11))$.  For instance,  $\exp(\i 2\theta \mathbb{X}(00,11))$ can be realized as
$(S\otimes S) \exp(-\i \theta \sigma_x\otimes \sigma_x)  (S^\dag\otimes S^\dag)\exp(\i \theta \sigma_x\otimes \sigma_x) $.

In this example, the elements of $\mathcal{SV}_{n,k}^{\mathbb{Z}_2}$ are  $n$-qubit unitaries $V$ satisfying  the conditions 
\be\label{con1}
\sigma_z^{\otimes n}V \sigma_z^{\otimes n}=V\ ,
\ee
and 
\be\label{con2}
\text{det}(V_0)=\text{det}(V_1)=1 \  ,
\ee
where  $V_0$ and $V_1$ are the components of $V$ in the subspaces with 
even and odd 
Hamming weights, respectively, such that $V=V_0\oplus V_1$ and $\text{det}(V_b)$ denotes the determinant of $V_b$ on its support.   In summary, we conclude that 

\begin{corollary}
Any $n$-qubit unitary $V$  satisfying conditions in Eq.(\ref{con1}) and  Eq.(\ref{con2})  can be implemented  using $S$ gates, 2-qubit gates in Eq.(\ref{CZ4}), and the family of unitaries $\exp(\i \theta \sigma_x\otimes \sigma_x): \theta\in[0,2\pi)$.  
\end{corollary}

\section{From semi-universality to universality via an ancilla qudit}\label{Sec:anc}

Ref. \cite{marvian2022restrictions} proves that any U(1)-invariant unitary on qubit systems can be realized with 2-local U(1)-invariant unitaries and a single ancillary qubit. Here, we show how this method  can be generalized to other Abelian symmetries to circumvent  type \textbf{I} constraints. 
 It is worth noting that, in general, type \textbf{II} constraints cannot be removed with ancillary qudits of the same size (For instance, in the example of the cyclic group $\mathbb{Z}_p$, it is impossible to convert state $|1\rangle^{\otimes p}$ to $|0\rangle^{\otimes p}$ with $k$-local symmetric unitaries with $k<p$, even if one is allowed to use  ancillary qubits).

Indeed, we establish a more general  result demonstrating the power of a single ancillary qudit. First, we show

\begin{lemma}\label{lem:anc}
For an Abelian group $G$, any unitary transformation $V$ that is diagonal in the basis  $\pmb{B}_n$ in Eq.(\ref{form}) can be realized with 2-local $G$-invariant unitaries and a single ancilla qudit. That is, there exists $\widetilde{V}\in\mathcal{V}^G_{n+1, 2}$, such that
 for any state $|\psi\rangle\in(\mathbb{C}^d)^{\otimes n}$, it holds that
 \be\nonumber
 \widetilde{V} (|\psi\rangle\otimes |0\rangle_{\text{anc}})= (V|\psi\rangle)\otimes |0\rangle_{\text{anc}}\ ,
 \ee
where  $|0\rangle\in\mathbb{C}^d$ can be any state that belongs to a single irrep of $G$, i.e., has a definite charge. 
\end{lemma} 
Note that even in 
  the presence of  
type \textbf{II} constraints, 
 all diagonal unitaries can be still implemented with 2-qudit gates and a single ancilla qudit. This is expected from the intuition that type \textbf{II} constraints are related to restrictions on how the conserved charge associated to the symmetry can be redistributed in the system and  implementing diagonal unitaries does not require such charge redistributions.


As we argue in section \ref{Sec:proof},  each basis element $|\pmb{r}\rangle\in \pmb{B}_n$ belongs to a single irreducible invariant subspace $\mathcal{H}_{\mu,\alpha}$. Then, it is straightforward to see that combining these diagonal unitaries with unitaries in 
$\bigoplus_{\mu,\alpha} \text{SU}\left(\mathcal{H}_{\mu,\alpha}\right) \subset \mathcal{V}_{n,k}^G$ in Eq.(\ref{subgroup}), one obtains all the unitaries that are block-diagonal with respect to the 
irreducible invariant subspaces $\{\mathcal{H}_{\mu,\alpha}\}$. In other words, for $k\ge 2$, a single ancillary qudit allows us to extend the group $\mathcal{V}_{n,k}^G$ to 
\be\nonumber
\mathcal{V}_{n,k}^G\ \   \xrightarrow{\text{one ancillary qudit}} \ \ \ \bigoplus_{\mu,\alpha} \text{U}\left(\mathcal{H}_{\mu,\alpha}\right)\ .
\ee
  If there are no extra conserved observables, i.e., Eq.(\ref{Abs}) holds, then one obtains $\mathcal{V}_{n,n}^G$, the set of all $G$-invariant unitaries on the system.

In the following we show that lemma \ref{lem:anc}  follows from theorem \ref{Thm2}. To simplify the notation and without loss of generality, we assume the state $|0\rangle$ of ancilla  is an element of $\pmb{B}_1$, the qudit  basis that was used in the definition of  in Eq.(\ref{form}).  Consider   the family of unitaries 
$\exp(\i\theta |\pmb{r}\rangle\langle \pmb{r}|): \theta\in[0,2\pi)$, where 
$|\pmb{r}\rangle$ is an arbitrary element of the basis $\pmb{B}_n$, other than $|0\rangle^{\otimes n}$.  Note that these unitaries are not in $\bigoplus_{\mu,\alpha} \text{SU}(\mathcal{H}_{\mu,\alpha})$, and therefore, in general, they are not realizable with $k$-local $G$-invariant unitaries with $k<n$ (This is the case, for instance, for U(1) symmetry). Since $|\pmb{r}\rangle \neq |0\rangle^{\otimes n}$, there exists, at least, one qudit with the reduced state $|s\rangle$, which is orthogonal to $ |0\rangle$.  
 Let  $|\pmb{r}'\rangle\in \pmb{B}_n$  be the state obtained from $|\pmb{r}\rangle$ by changing the state of this qudit from $|s\rangle$ to $|0\rangle$.   This means that the two states  
   \be\nonumber
 |\pmb{r}\rangle\otimes |0\rangle_{\text{anc}}  \ \ \ \ \xleftrightarrow{\text{\ \ swap \  \ }} \ \ \ \ \ |\pmb{r}'\rangle \otimes |s\rangle_{\text{anc}}
 \ee
are related by swapping the ancilla qudit and the qudit in state $|s\rangle$. Since 
swap is 2-local and $G$-invariant, this implies that, relative to 
the group $\mathcal{V}^G_{n+1,2}$,  these states live in the same 
irreducible invariant subspace of $(\mathbb{C}^d)^{\otimes (n+1)}$, as defined in theorem \ref{Thm2}. It follows that the support of the  Hamiltonian
 \be\nonumber
\widetilde{H}_{\pmb{r}}=|\pmb{r}\rangle\langle \pmb{r}|\otimes |0\rangle\langle 0|_{\text{anc}}-|\pmb{r}'\rangle\langle \pmb{r}'|\otimes |s\rangle\langle s|_{\text{anc}}
\ee
is restricted to a single irreducible invariant subspace. Since this Hamiltonian is also traceless, theorem \ref{Thm2} implies that 
\be\nonumber
\forall\theta\in[0,2\pi):\ \ \ \ \ \  \exp(i \theta \widetilde{H}_{\pmb{r}}) \in \mathcal{V}^G_{n+1,2}\ .
\ee
Furthermore, we have 
\be\nonumber
\exp(\i \theta \widetilde{H}_{\pmb{r}}) (|\psi\rangle\otimes |0\rangle_{\text{anc}})= \big[\exp(\i \theta |\pmb{r}\rangle\langle \pmb{r}|)  |\psi\rangle\big]\otimes |0\rangle_{\text{anc}}\ .
\ee
Since the state of ancilla has remained unchanged, we can reuse it again. Using this method, we can  implement unitaries $\exp(\i \theta |\pmb{r}\rangle\langle \pmb{r}|)$ for all $|\pmb{r}\rangle\in\pmb{B}_n$, except $|\pmb{r}\rangle=|0\rangle^{\otimes n}$.  Combining these unitaries we  obtain any unitary in the form
\be\nonumber
V=|0\rangle\langle 0|^{\otimes n}+\sum_{\pmb{r}\neq 0^n} \exp(\i \theta_{\pmb{r}})\  |\pmb{r}\rangle\langle \pmb{r}| \ ,
\ee
for arbitrary phases $\{\exp(\i \theta_{\pmb{r}})\}$. This is the set of all  unitaries diagonal in $\pmb{B}_n$ basis, up to a global phase. Adding a global phase,  which by definition is  $k$-local for all $k$, we obtain all unitaries that are diagonal in this basis. This proves lemma \ref{lem:anc}.

\section{ Discussion}\label{Sec:discuss}


Since the early works of Deutsch \cite{deutsch1985quantum}, DiVincenzo \cite{divincenzo1995two}, and Lloyd \cite{lloyd1995almost}, the theory of quantum circuits has significantly advanced. However, despite their diverse applications, symmetric quantum circuits are not yet well-understood. 
In this work we developed the theory of Abelian symmetric quantum circuits and showed that certain constraints  that restrict realizable unitaries in circuits with non-Abelian symmetries do not exist in the case  of circuits with 
Abelian symmetries (namely, constraints of types \textbf{III} and \textbf{IV}). In particular, 
according to theorem \ref{Thm2}, the realized unitaries in different irreducible invariant subspaces $\{\mathcal{H}_{\mu,\alpha}\}$ are, in general, independent of each other, and all the unitaries with determinant one are realizable inside each subspace. 
It is also worth noting that 
the constraints of types \textbf{III} and \textbf{IV},  do not always appear in systems with non-Abelian symmetries.  This is the case for qubit systems with SU(2) symmetry \cite{marvian2022rotationally}. On the other hand, these constraints exist for qudit circuits with SU($d$) symmetry with $d\ge 3$ \cite{Marvian2021qudit}.


Type \textbf{I} constraints  can be characterized using the methods developed in \cite{marvian2022restrictions}, which are briefly reviewed in Appendix \ref{Sec:App1}. As it can be seen from Theorem \ref{Thm4} and lemma \ref{lem0},  type  \textbf{II} constraints are related to the fact that the locality of interactions restrict redistribution of the total charge in the system.  In particular, these constraints are  determined by the set of irreps (charges) carried by a single qudit and those appearing in the total system.  The irreps of any Abelian group $G$  themselves form an Abelian group, called the Pontryagin dual group $\hat{G}$ \cite{deitmar2014principles}. When $G$ is a finite or compact Lie group, which is the case of interest in this paper, $\hat{G}$ is a discrete group. Furthermore, if $G$ is connected, then $\hat{G}$ is torsion-free, which means that, unlike the case of $\mathbb{Z}_p$ example, copies of the same charge cannot annihilate each other.  

 We saw that in the case of finite Abelian groups,  there is a finite $k\le |G|$, independent of system size $n$, such that $k$-local symmetric unitaries can arbitrarily redistribute the charge in the system, which by theorem \ref{Thm4} implies semi-universality.  For many other groups of interest in physics,   the dual Pontryagin group is finitely generated, which essentially means there are a finite type of independent charges in the system. In this situation again one can show that there is a finite $k$, independent of the system size, such that $k$-local $G$-invariant unitaries become semi-universal.


An interesting future direction is exploring the connection between Abelian circuits and their classical counterparts in terms of reversible logic circuits. In general, from the point of view of universality, there is no immediate relation among the classical and quantum circuits. For example, while in the absence of symmetries 2-local gates are universal for quantum circuits, in classical reversible circuits  universality requires 3-local gates (e.g., the Toffoli gate). Nevertheless, given that type \textbf{II}  constraints in Abelian circuits have a simple classical interpretation, it will be interesting to further study possible connections of these circuits with  classical reversible circuits with conservation laws (See, e.g.,  \cite{aaronson2015classification}).  

\subsection*{Applications and implications}

Symmetric unitaries are ubiquitous across quantum information science. Many protocols and algorithms involve 
symmetric unitaries.  Symmetric quantum circuits have been shown to be useful useful in the context of  variational quantum machine learning
\cite{meyer2023exploiting, nguyen2022theory,  zheng2023speeding},  variational quantum eigensolvers  \cite{barron2021preserving, shkolnikov2021avoiding, gard2020efficientsymmetry},  and quantum gravity \cite{nakata2023black}.
 Besides,  symmetric quantum circuits have  become a standard framework for understanding the phases and  dynamics of many-body systems, e.g., in the context of classification of symmetry-protected topological phases \cite{chen2010local, chen2011classification} and quantum chaos in the presence of symmetries \cite{khemani2018operator}.

Hence,  understanding what class of unitary transformations can be realized with such circuits is a basic question with broad relevance.  Here, we briefly discuss some implications and applications of our results. \\

\noindent\textbf{Thermalization in the presence of non-Abelian conserved charges.} In the recent years scrambling and thermalization of closed quantum systems in the presence of conserved charges have been extensively studied. Researchers have  also considered possible implications of the presence of non-Abelian conserved charges in this context (See, e.g., \cite{majidy2023non,  yunger2022build, kranzl2022experimental}).  In particular, it has been conjectured that the non-commutativity of conserved charges hinders the  thermalization of quantum systems \cite{halpern2020noncommuting}.

We observed  that, in the presence of non-Abelian charges, the locality of interactions can obstruct the dynamics of the system in a way that cannot happen for Abelian symmetries (Note that although our theorems  are phrased in the language of quantum circuits, as explained in section \ref{sec:setup}, they can be equivalently  understood in terms of continuous Hamiltonian time evolution of a closed system with a global symmetry).  
In particular, we showed that due to the additivity of Abelian charges,  the time evolution in different charge sectors can be independent (decoupled) from each other, whereas in general  this  cannot happen in the case of non-Abelian charges.  Therefore, our result supports the intuition that non-commutativity of conserved charges may, in some sense, slow down the  thermalization process.  A better understanding of this phenomenon and establishing  the  connections with thermalization requires further investigation. \\


\noindent\textbf{Quantum thermodynamics and the resource theory of  asymmetry.} Symmetric unitaries play a central  role in quantum thermodynamics \cite{janzing2000thermodynamic, FundLimitsNature, brandao2013resource, guryanova2016thermodynamics, lostaglio2015quantumPRX, halpern2016microcanonical, halpern2016beyond}, the resource theory of asymmetry \cite{gour2008resource, Marvian_thesis, marvian2013theory}, and the closely related topic of covariant error correcting codes  \cite{faist2020continuous, hayden2021error, kong2021charge}.  In thermodynamics it is often assumed that  energy-conserving unitaries can be realized with negligible thermodynamic costs.  The same assumption is made when there are additional conserved  charges. Similarly, the resource theory of asymmetry focuses on operations that are realizable with symmetric unitaries and symmetric ancillary systems. 


But,  how can we implement a general symmetric unitary on a composite system? Fundamental laws of nature as well as practical restrictions limit us to local interactions, i.e., those that couple a few  subsystems together. Can we realize a general symmetric unitary by combining local symmetric unitaries? Or, perhaps, realizing a general symmetric unitary on a composite system requires symmetry-breaking interactions. Our results in this paper show that, at least, in the case of Abelian symmetries, all symmetric unitaries are resizable with local symmetric unitaries, and a single ancilla qudit, which justifies the fundamental assumptions of the resource theories of thermodynamics and asymmetry for such symmetries.\\  

\noindent\textbf{Subspace controllability and  quantum computing in decoherence-free subspaces.} The problem  of realizing symmetric unitaries in a single charge sector has been previously considered  in the context of universal quantum computing in decoherence-free subspaces. In particular, researchers have studied implementing  certain U(1)-invariant unitaries with $XX+YY$ interaction, as well as implementing certain  SU(2)-invariant unitaries with the  Heisenberg exchange  interaction in a subspace with one irrep of the symmetry \cite{Bacon:Sydney, Kempe:01, Bacon:2000qf, DiVincenzo:2000kx, kempe2001theory, bacon2001coherence, levy2002universal, rudolph2005relational, viola2000dynamical, rudolph2021relational}.  It should be noted that here the focus is often on demonstrating the computational universality, also known as the encoded universality, which does not require implementing all symmetric unitaries in the subspace.

  Another closely related topic is subspace controllability, which requires that all unitary transformations on a subspace of the Hilbert space should be realizable  (See, e.g.,\cite{wang2016subspace, albertini2021subspace}).  In this context, it has been shown that Hamiltonians $XX+YY$, $ZZ$, and local $Z$ generate all unitaries in a single irrep of U(1) symmetry \cite{wang2016subspace}.

In this work, on the other hand, we studied a stronger notion of controllability, namely, semi-universality. Unlike subspace controllability, which focuses on a single charge subspace, semi-universality requires implementing the desired symmetric unitaries in all charge sectors, with possible constraints on the relative phases between sectors. Of course, this stronger notion implies subspace controllability on a single charge sector $\mathcal{H}_\mu$ of a global symmetry. In particular, Theorem \ref{Thm4} implies that in the case of Abelian symmetries, the only possible obstruction to subspace controllability is the existence of irreducible invariant subspaces. These can be easily identified either by characterizing the commutant of the control Hamiltonians in the subspace under consideration or, alternatively, by using the notion of transitivity in Statement 3 of Theorem \ref{Thm4} (see also Lemma \ref{lem0}).  More generally, theorem \ref{Thm2} implies that on each irreducible invariant subspace $\mathcal{H}_{\mu,\alpha}$ the subspace controllability can be achieved with $k$-local symmetric Hamiltonians. In summary, this approach  reveals a deeper understanding of the subspace controllability, connecting it to the symmetries and locality of the control Hamiltonians.

\section{Methods}\label{Sec:proof}
In the following, first we prove theorem \ref{Thm2} and then, using this theorem,  in section \ref{secd} we prove theorem \ref{Thm4}, which  also implies theorem \ref{Thm3}.

\subsection{Proof of theorem   \ref{Thm2}}
The subspaces $\{\mathcal{H}_{\mu,\alpha}\}$ in the statement of  theorem 
\ref{Thm2} 
can be defined and characterized based on  
$\text{Comm}(\mathcal{V}^G_{n,k})$. First, using the fact  that 
\be
\mathcal{V}^G_{n,1}\subseteq  \mathcal{V}^G_{n,k}\subseteq  \mathcal{V}^G_{n,n} \ ,
\ee
we have
\be\label{ComCom}
\text{Comm}(\mathcal{V}^G_{n,1})\supseteq  \text{Comm}(\mathcal{V}^G_{n,k})\supseteq  \text{Comm}(\mathcal{V}^G_{n,n}) \  .
\ee
Applying Schur's lemma, or the bicommutant theorem, one can easily see  that the only operators commuting with all symmetric unitaries are projectors to subspaces $\{\mathcal{H}_\mu\}$ and their linear combinations, i.e., 
\be
\text{Comm}(\mathcal{V}^G_{n,n})=\text{Span}_\mathbb{C}\{\Pi_\mu\}\ . 
\ee
Next, we focus on $\text{Comm}(\mathcal{V}^G_{n,1})$. Since $\mathcal{V}^G_{n,1}$  contains single-qudit unitary $\exp(\i \theta |r\rangle\langle r|)$ tensor product with the identity operators on the other $n-1$ qudits,  for arbitrary $|r\rangle\in \pmb{B}_1$ and $ \theta\in[0,2\pi)$,  any operator in $\text{Comm}(\mathcal{V}^G_{n,1})$  commutes with $|r\rangle\langle r|$  on this qudit. It follows that all operators in $\text{Comm}(\mathcal{V}^G_{n,1})$ are diagonal in basis $\pmb{B}_n$. Combining this with Eq.(\ref{ComCom}), we conclude that for all $k\ge 1$ it holds that
\be\label{commutant}
\text{Span}_\mathbb{C}\{\Pi_\mu\}\subseteq 
\text{Comm}(\mathcal{V}^G_{n,k}) \subseteq    \text{Span}_\mathbb{C}\{|\pmb{r}\rangle\langle \pmb{r}|: |\pmb{r}\rangle\in \pmb{B}_n \} \ .
\ee
This, in particular, means that all operators in $\text{Comm}(\mathcal{V}^G_{n,k})$ commute with each other.  The common eigen-subspaces of  operators in  $\text{Comm}(\mathcal{V}^G_{n,k}) $ decompose the total Hilbert space into  orthogonal subspaces  $\{\mathcal{H}_{\mu,\alpha}\}$, as in Eq.(\ref{Thm2}), with the property that on each 
subspace $\mathcal{H}_{\mu,\alpha}$ any operator in $\text{Comm}(\mathcal{V}^G_{n,k}) $ takes a constant value.  In other words,
\be\label{contra}
\text{Comm}(\mathcal{V}^G_{n,k}) = \text{Span}_\mathbb{C}\{\Pi_{\mu,\alpha}\}\subseteq \text{Span}_\mathbb{C}\{|\pmb{r}\rangle\langle \pmb{r}|: |\pmb{r}\rangle\in \pmb{B}_n \}\ , 
\ee
 where $\Pi_{\mu,\alpha}$ is the projector to $\mathcal{H}_{\mu,\alpha}$. Note that because projectors $\{\Pi_\mu\}$ are in $\text{Comm}(\mathcal{V}^G_{n,k})$, each subspace $\mathcal{H}_{\mu,\alpha}$ is contained in a single charge sector $\mathcal{H}_\mu$, which justifies the label  $\mu,\alpha$ of subspace $\mathcal{H}_{\mu,\alpha}$. That is,
 \be
 \mathcal{H}_\mu=\bigoplus_\alpha \mathcal{H}_{\mu,\alpha}\ .
 \ee
We conclude that under the action of $\mathcal{V}^G_{n,k}$  the orthogonal  subspaces $\{\mathcal{H}_{\mu,\alpha}\}$ are invariant, which means
\be\label{inclusion}
\mathcal{V}^G_{n,k} \subseteq \bigoplus_{\mu,\alpha} \text{U}\left(\mathcal{H}_{\mu,\alpha}\right)\ .
\ee
In the rest of this proof we show that $\mathcal{V}^G_{n,k} $ contains the subgroup 
 $\bigoplus_{\mu,\alpha} \text{SU}\left(\mathcal{H}_{\mu,\alpha}\right)$. 
 
First, note that $\mathcal{V}^G_{n,k}$ acts irreducibly on each subspace $\mathcal{H}_{\mu,\alpha}$. Otherwise, by Schur's lemma, there exists an operator commuting with $\mathcal{V}^G_{n,k}$,  with support restricted to $\mathcal{H}_{\mu,\alpha}$ which is not proportional to the projector $\Pi_{\mu,\alpha}$. But, this contradicts Eq.(\ref{contra}). 

Next, recall that each projector  $\Pi_{\mu,\alpha}$ is diagonal in basis  $\pmb{B}_n$,  which means state 
 $|\pmb{r}\rangle\in \pmb{B}_n$ is a vector in a single subspace $\mathcal{H}_{\mu,\alpha}$. Furthermore,


\begin{lemma}\label{lem0}
A pair of basis elements  $|\pmb{r}\rangle, |\pmb{r}'\rangle\in \pmb{B}_n$ belong to the same   irreducible invariant subspace $\mathcal{H}_{\mu,\alpha}$ if, and only if,   there exists a sequence of  elements of  $\pmb{B}_n$ connecting 
$|\pmb{r}\rangle$ to $|\pmb{r}'\rangle$  as 
\be\label{seq}
|\pmb{r}\rangle=|\pmb{s}^1\rangle \longrightarrow  |\pmb{s}^2\rangle \longrightarrow  \cdots \cdots \cdots\longrightarrow  |\pmb{s}^t\rangle=|\pmb{r}'\rangle\ ,
\ee
such that  (i) all states $|\pmb{s}^j\rangle$ are in the same irreducible subspace $\mathcal{H}_{\mu,\alpha}$, and (ii)  the Hamming distance between any consecutive pair is $d(\pmb{s}^{j}, \pmb{s}^{j+1})\le k$.   
\end{lemma}
 \begin{proof}

If $|\pmb{r}\rangle$ and $|\pmb{r}'\rangle$ are in the same  irreducible invariant  subspace, then there exists $V\in\mathcal{V}^G_{n, k}$ such that
$\langle \pmb{r}'| V|\pmb{r}\rangle\neq 0$. Then, since  group $\mathcal{V}^G_{n,k}$ is generated  by $k$-local $G$-invariant  unitaries and is compact,  this unitary has a decomposition as  $V=V_{t-1} \cdots V_1$, where $V_1, V_2, \cdots, V_{t-1}$  is a finite sequence of $k$-local $G$-invariant unitaries (Recall that compact groups are uniformly finitely generated \cite{d2007introduction}). This implies  
\be\nonumber
\langle\pmb{r}'|V  |\pmb{r}\rangle=\langle\pmb{r}'|V_{t-1}\cdots V_2V_1 |\pmb{r}\rangle \neq 0\ .
\ee 
Since $\pmb{B}_n$  is an orthonormal basis,   $\sum_{|\pmb{s}\rangle\in \pmb{B}_n}|\pmb{s}\rangle\langle \pmb{s}|$ is the identity operator. Inserting this resolution of the identity  in the above equation, we find that there exists a sequence of basis elements  $|\pmb{s}^2\rangle, |\pmb{s}^3\rangle, \cdots, |\pmb{s}^{t-1}\rangle$ such that     
\begin{align}\nonumber
\langle\pmb{r}'|V_{t-1}|\pmb{s}^{t-1}\rangle\langle \pmb{s}^{t-1}|  \cdots  \cdots  |\pmb{s}^3\rangle\langle\pmb{s}^3|V_2  |\pmb{s}^2\rangle \langle\pmb{s}^2|V_1 |\pmb{r}\rangle \neq 0\ .
\end{align}
Defining $|\pmb{s}^1\rangle=|\pmb{r}\rangle$ and $|\pmb{s}^{t}\rangle=|\pmb{r}'\rangle$, this  implies 
\be\label{c2}
\langle\pmb{s}^{j+1}|V_j|\pmb{s}^j\rangle\neq 0\  \ \  :\ j=1, \cdots, {t-1}\ .
\ee
Unitary $V_j$ is $k$-local and therefore acts non-trivially on, at most,  $k$ qudits. Then, on the rest of $n-k$ qudits two states $|\pmb{s}^j\rangle$ and $|\pmb{s}^{j+1}\rangle$
should be identical, because otherwise the left-hand of Eq.(\ref{c2}) will be zero.  This implies that the Hamming distance $d(\pmb{s}^{j}, \pmb{s}^{j+1})\le k$. Finally, note that because each element of basis $\pmb{B}_n$  is restricted to a single irreducible invariant subspace, and $k$-local unitaries are block-diagonal with respect to these subspaces, then all states $|\pmb{s}^1\rangle, |\pmb{s}^2\rangle, |\pmb{s}^3\rangle, \cdots, |\pmb{s}^{t-1}\rangle, |\pmb{s}^t\rangle$ are in the same irreducible invariant subspace. This proves one direction of the lemma. The proof of the other direction is straightforward and can be found in the proof lemma \ref{lem4} below. 

\end{proof} 

Next, we apply lemma \ref{lem0} to show that for any pair of distinct basis elements  
\be\nonumber
|\pmb{r}\rangle, |\pmb{r}'\rangle\in \mathcal{H}_{\mu,\alpha}\  ,
\ee
in the same irreducible invariant subspace $\mathcal{H}_{\mu,\alpha}$,  the unitary time evolutions generated by 
 Hamiltonians
 \begin{align}\nonumber
\mathbb{X}({\pmb{r}\ ; \pmb{r}'})&= |\pmb{r}\rangle\langle \pmb{r}'|+|\pmb{r}'\rangle\langle \pmb{r}|\ ,\nonumber\\
\mathbb{Y}({\pmb{r}\ ;\pmb{r}'})&= \i(|\pmb{r}\rangle\langle \pmb{r}'|-|\pmb{r}'\rangle\langle \pmb{r}|)\ ,\nonumber\\ 
 \ \mathbb{Z}({\pmb{r}\ ;\pmb{r}'})&=\frac{\i}{2} [\mathbb{X}({\pmb{r}\ ;\pmb{r}'}), \mathbb{Y}({\pmb{r},\pmb{r}'})]= |\pmb{r}\rangle\langle \pmb{r}|-|\pmb{r}'\rangle\langle \pmb{r}'|\ ,\nonumber
 \end{align}
 can be realized with $k$-local $G$-invariant unitaries. That is  
\begin{align}\label{time2}
 &\forall \theta\in[0,2\pi):  \\ &\exp[\i \theta \mathbb{X}({\pmb{r}\ ; \pmb{r}'})] \ ,   \exp[\i\theta \mathbb{Y}({\pmb{r}\ ; \pmb{r}'})]\ ,  \exp[\i\theta \mathbb{Z}({\pmb{r}\ ; \pmb{r}'})]\in \mathcal{V}^G_{n, k}\  .\nonumber
 \end{align}
The first step is to show that this holds for the special case where the Hamming distance $d({\pmb{r},\pmb{r}'})\le k$.

 \begin{lemma}\label{lem4}
Suppose two distinct basis elements $|\pmb{r}\rangle, |\pmb{r}'\rangle\in \pmb{B}_n$ belong to the same charge sector $\mathcal{H}_\mu$ and their   Hamming distance is  $d(\pmb{r},\pmb{r}')\le k$. Then, the statement in Eq.(\ref{time2}) holds for $k\ge 2$.
 \end{lemma}
We present the proof of this lemma at the end.  Next, we use the unitaries in lemma \ref{lem4} as building blocks to realize unitaries in Eq.(\ref{time2}) in the general case, where the Hamming distance $d(\pmb{r},\pmb{r}')$ is not bounded by $k$. To  this end, we note that for any 3 distinct states $|\pmb{s}\rangle, |\pmb{s}'\rangle, |\pmb{s}''\rangle$ in basis $\pmb{B}_n$, 
 \be\label{commutq}
  \big[\mathbb{X}(\pmb{s}\ ; \pmb{s}')\ ,\ \mathbb{Y}(\pmb{s}'\ ; \pmb{s}'')\big]=\i \mathbb{X}(\pmb{s}\ ; \pmb{s}'') \ .
 \ee
This identity implies that if $\exp(\i\theta \mathbb{X}(\pmb{s} ;\pmb{s}'))$ and $\exp(\i\theta Y(\pmb{s}' ; \pmb{s}''))$ are in the group $\mathcal{V}^G_{n, k}$ for all $\theta\in[0,2\pi)$, then unitaries  $\exp(\i\theta \mathbb{X}(\pmb{s} ; \pmb{s}''))$ are also in this group for all $\theta$. Similar constructions work for $\mathbb{Y}(\pmb{s} ;  \pmb{s}'')$ and $\mathbb{Z}(\pmb{s} ; \pmb{s}'')$ (Recall that the set of realizable Hamiltonians are closed under the commutator operation and, hence, form a Lie algebra \cite{d2007introduction}).

Next, we apply this result recursively to the sequence of states obtained in Eq.(\ref{seq}). 
 Since any  consecutive pair of states $|\pmb{s}^j\rangle$ and $|\pmb{s}^{j+1}\rangle$ belong to the same irreducible invariant  subspace  and   $d(\pmb{s}^{j},\pmb{s}^{j+1})\le k$, applying  lemma \ref{lem4}, we find that Hamiltonians 
 $$\mathbb{X}(\pmb{s}^{j} ; \pmb{s}^{j+1})\ ,\  \mathbb{Y}(\pmb{s}^{j} ; \pmb{s}^{j+1})\ , \text{and     }\ \  \mathbb{Z}(\pmb{s}^{j} ; \pmb{s}^{j+1})$$
  can be realized with $k$-local $G$-invariant unitaries.  Then, applying Eq.(\ref{commutq}) together with lemma \ref{lem0}, we conclude that 
for all pair of basis elements $|\pmb{r}\rangle$ and $|\pmb{r}'\rangle$ in the same irreducible invariant subspace $\mathcal{H}_{\mu,\alpha}$, Eq.(\ref{time2}) holds. 
Such unitaries that act non-trivially only in a subspace spanned by 2 basis elements are called 2-level unitaries \cite{NielsenAndChuang} (also known as Givens rotations) and  they generate the full special unitary group on the space \cite{NielsenAndChuang}.  
 Since $\mathbb{X}({\pmb{r} ; \pmb{r}'})$, $\mathbb{Y}({\pmb{r} ; \pmb{r}'})$, and $ \mathbb{Z}({\pmb{r} ;\pmb{r}'})$ are traceless, it follows that the group generated by these unitaries is $\text{SU}(\mathcal{H}_{\mu,\alpha})$ on $\mathcal{H}_{\mu,\alpha}$   
(This can   also be seen using the fact that operators $\mathbb{X}({\pmb{r}\ ; \pmb{r}'})$, $\mathbb{Y}({\pmb{r} ; \pmb{r}'})$, and $ \mathbb{Z}({\pmb{r} ;\pmb{r}'})$ with ${\pmb{r},\pmb{r}'}\in \mathcal{H}_{\mu,\alpha}$ span the space of traceless Hermitian operators on  $\mathcal{H}_{\mu,\alpha}$). We emphasize that these unitaries act non-trivially only in a single  irreducible invariant subspace $\mathcal{H}_{\mu,\alpha}$.   Therefore, the unitary transformations realized in different sectors remain independent of each other. This proves that $\mathcal{V}^G_{n,k}$ contains the subgroup $\bigoplus_{\mu,\alpha} \text{SU}\left(\mathcal{H}_{\mu,\alpha}\right)$. To finish the proof of theorem \ref{Thm2}, in the following we prove  lemma \ref{lem4}.

\begin{proof}(Lemma \ref{lem4})  For $k=n$, the statement holds trivially. Hence, we assume $k<n$ in the following.  We partition the qudits in the system into two subsystems $A$ and $B$, where $A$ is the set of qudits 
whose assigned reduced states are identical for states $|\pmb{r}\rangle$ and  $|\pmb{r}'\rangle$, and $B$ contains the rest of qudits. By assumption, the Hamming distance  $d(\pmb{r}, \pmb{r}')\le k$, which means the number of qudits in $B$ is $|B|\le k$. Relative to this partition, the two states are decomposed as  
$$|\pmb{r}\rangle=|\pmb{r}_A\rangle_A\otimes |\pmb{r}_B\rangle_B\ \  \ \text{, and}  \ \   \ \  |\pmb{r}'\rangle=|\pmb{r}_A\rangle_A\otimes |\pmb{r}'_B\rangle_B\ .$$
 Then, 
\be\label{art}
\mathbb{X}(\pmb{r}, \pmb{r}')= |\pmb{r}_A\rangle\langle \pmb{r}_A|_A \otimes  \mathbb{X}_B(\pmb{r}_B, \pmb{r}'_B) \ ,
\ee
where  $\mathbb{X}_B(\pmb{r}_B, \pmb{r}'_B)=|\pmb{r}_B\rangle\langle \pmb{r}'_B|_B+|\pmb{r}'_B\rangle\langle \pmb{r}_B|_B$ is an operator acting on qudits in $B$.

Next, we focus on the Hamiltonian $ \mathbb{I}_A \otimes \mathbb{X}_B(\pmb{r}_B, \pmb{r}'_B) $, where $\mathbb{I}_A$ is the identity operator on $A$. Since  $|B|\le k$,  this Hermitian operator is $k$-local. Furthermore, it is $G$-invariant. Roughly speaking, this holds because  by the assumption of the lemma, in states $|\pmb{r}\rangle$ and $|\pmb{r}'\rangle$  the total charge in the system is equal, and also the total charge in subsystem $A$ is  equal, which in turn implies the total charge in subsystem $B$ is equal. More formally, we have 
\be\nonumber
\forall g\in G:\ \ \ \  \langle\pmb{r}|U(g) |\pmb{r}\rangle=\langle\pmb{r}'|U(g) |\pmb{r}'\rangle=e^{i\mu(g)}\ ,
\ee
where $U(g)=u(g)^{\otimes n}$, and the phase $e^{i\mu(g)}$ is an irrep of group $G$.  This immediately implies 
\be\nonumber
\forall g\in G:\  \langle\pmb{r}_B|u(g)^{\otimes |B|} |\pmb{r}_B\rangle=\langle\pmb{r}'_B|u(g)^{\otimes |B|} |\pmb{r}'_B\rangle=e^{i\mu_B(g)}\ ,
\ee
where 
the phase $e^{i\mu_B(g)}$ is also an irrep of group $G$.  This, in turn, implies $\mathbb{X}_B(\pmb{r}_B, \pmb{r}'_B)$ commutes with $u(g)^{\otimes |B|}$ for all $g\in G$.
The fact that $\mathbb{I}_A \otimes \mathbb{X}_B(\pmb{r}_B, \pmb{r}'_B)$ is $k$-local and $G$-invariant means that  the unitaries $\mathbb{I}_A \otimes  \exp(i\theta \mathbb{X}_B(\pmb{r}_B, \pmb{r}'_B))$ are in $\mathcal{V}_{n,k}^G$  for all $\theta\in[0,2\pi)$.

 Next, using the family of unitaries $\exp(i\theta \mathbb{X}_B(\pmb{r}_B, \pmb{r}'_B))$, we construct the desired unitary $\exp(i\theta \mathbb{X}(\pmb{r}, \pmb{r}'))$ and show that it is also in $\mathcal{V}_{n,k}^G$. We use a construction which is analogous  to the circuit identity in Fig.  \ref{Fig}, where the unitary $\exp(i\theta \mathbb{X}_B(\pmb{r}_B, \pmb{r}'_B))$ plays the role of $R=\exp(i\theta \sigma_x)$ in this circuit.

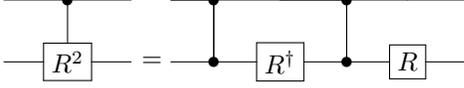
\begin{figure}
\[
\Qcircuit @C=1.5em @R=1.5em {
  & \ctrl{1} &\qw  &&  \ctrl{1} \qw    & \qw&\ctrl{1} \qw    & \qw& \qw & \\
 & \gate{R^2} &\qw & \hspace{-5mm}= & \ctrl{-1}&  \gate{R^\dag} &  \ctrl{-1}  &  \gate{R}  & \qw &    }
\large
\]
\caption{The arguments used in the proof of lemmas \ref{lem4}  and  \ref{lem651} are  a generalization of the standard circuit identity in this figure, where  $R=\exp(i\theta\sigma_x)$ for  
$\theta\in[0,2\pi)$, and the 2-qubit gate in the right circuit is the Controlled-Z gate. }\label{Fig}
\end{figure}

Suppose we label all qubits in $A$ as $l=1,\cdots, |A|$. Then, define $F_0=\mathbb{X}_B(\pmb{r}_B, \pmb{r}'_B)$, and
\begin{align}\label{Fl}
F_l=2^l  \Big[|r_{1}\rangle\langle r_1|_1\otimes \cdots \otimes  |r_{l}\rangle\langle r_l|_l\Big] \otimes \mathbb{X}_B(\pmb{r}_B, \pmb{r}'_B)\ ,
\end{align}
where we have suppressed the tensor products with the identity operators on other qudits.  Comparing this definition with Eq.(\ref{art}), we find that $F_{|A|}=2^{|A|} \mathbb{X}(\pmb{r},\pmb{r}')$. For any pair of 
 distinct qudits $a$ and $b$, and any pair of basis elements $|r_a\rangle,|r_b\rangle\in \pmb{B}_1$,  define the Hermitian unitary operator   
\begin{align}\label{Cab}
C_{a,b}=&\mathbb{I}_d^{\otimes n}-2|r_a\rangle\langle r_a|_{a}\otimes |r_b\rangle\langle  r_b|_b\ .
\end{align}
Note that this unitary is 2-local and $G$-invariant (This unitary is a generalization of the controlled-Z unitary in Fig. \ref{Fig}).  Choosing $b$ to be any qudit in subsystem $B$, 
and using the fact that $\langle  r_b|F_l|r_b\rangle_b=0$, one can easily check the recursive relation  
\begin{align}\label{rec71}
F_{l+1}= F_{l}- C_{l+1, b}\ F_l\ C_{l+1, b} \ .
\end{align}
Clearly, all operators $F_l:  l=1,\cdots, |A|$ commute with each other. Furthermore, since  $C_{l+1, b}\ F_l\ C_{l+1, b}=F_{l}-F_{l+1}$, this operator  also commutes with all $F_l$. Then, Eq.(\ref{rec71}) implies that for all $\theta\in [0,2\pi)$ it holds that
\begin{align}\label{eq26}
\exp(\i\theta F_{l+1})= \exp(\i\theta F_{l})\ C_{l+1, b}\ \exp(-\i\theta F_l)\ C_{l+1, b} \ ,
\end{align}
which is analogous to the circuit in Fig. \ref{Fig}. This implies that by combining unitaries $C_{lb}: l=1,\cdots , |A|$ and unitaries $\exp(\i\theta \mathbb{X}_B(\pmb{r}_B, \pmb{r}'_B) )$, which all belong to $ \mathcal{V}_{n,k}^G$,  one can obtain  unitaries 
$\exp(\i\theta F_l): l=1,\cdots, |A|$, and in particular, $\exp(\i\theta \mathbb{X}(\pmb{r}, \pmb{r}') )$, for all $ \theta\in[0,2\pi)$.  Finally, since
\be\label{lastEq}
\mathbb{Y}(\pmb{r}, \pmb{r}')= \exp(\frac{\i\pi }{2}|r_b\rangle\langle r_b|)\    \mathbb{X}(\pmb{r}, \pmb{r}')\  \exp(\frac{-\i\pi }{2}|r_b\rangle\langle r_b|)\ ,
\ee
by combining $\exp(\i\theta \mathbb{X}(\pmb{r}, \pmb{r}') )$
with single-qudit $G$-invariant unitaries, we obtain $\exp(\i\theta \mathbb{Y}(\pmb{r}, \pmb{r}') )$. Furthermore, since 3 operators  $\mathbb{X}(\pmb{r}, \pmb{r}')$ , $\mathbb{Y}(\pmb{r}, \pmb{r}')$, and $\mathbb{Z}(\pmb{r}, \pmb{r}')$ satisfy the standard commutation relations of  $\mathfrak{su}(2)$ satisfied by Pauli operators, 
by sandwiching $\exp(\i\theta \mathbb{X}(\pmb{r}, \pmb{r}') )$  between 
 $\exp(\i\pi/4 \mathbb{Y}(\pmb{r}, \pmb{r}') )$ 
and its inverse we obtain $\exp(\i\theta \mathbb{Z}(\pmb{r}, \pmb{r}') )$.
This completes the proof of lemma \ref{lem4}, and hence theorem \ref{Thm2}.
\color{black}
\end{proof}

For future use, we summarize the main step of the above argument in the following lemma.

\begin{lemma}\label{lem651}
Consider a pair of distinct basis elements $|\pmb{r}\rangle, |\pmb{r}'\rangle\in \pmb{B}_n$ with the property that  
\be
|\pmb{r}'\rangle=\big[\mathbb{X}_C(\pmb{s},\pmb{s}')\otimes \mathbb{I}_{\overline{C}} \big] |\pmb{r}\rangle\ ,
\ee
where $C$ is a subset of qudits containing $c\le n$ qudits, $|\pmb{s}\rangle, |\pmb{s}'\rangle\in \pmb{B}_{c}$, and  $\mathbb{X}_C(\pmb{s};\pmb{s}')=|\pmb{s}\rangle\langle \pmb{s}'|+ |\pmb{s}'\rangle\langle \pmb{s}|$ acts on qudits in $C$, and $\mathbb{I}_{\overline{C}}$ is the identity operator on the rest of qudits. Then, the family of unitaries $\exp(\i\theta \mathbb{X}(\pmb{r};\pmb{r}')): \theta\in[0,2\pi)$ can be realized using  $c$-local unitaries $\exp(\i\theta \mathbb{X}_C(\pmb{s};\pmb{s}')): \theta\in[0,2\pi)$ and 2-qudit gates $C_{a,b}$ in Eq.(\ref{Cab}).  
\end{lemma}
The proof follows from the above argument. In particular, to apply this construction 
first we determine which qudits in $C$ are associated with different  reduced states in states $|\pmb{s}\rangle$ and $|\pmb{s}'\rangle$. This subset of qudits, which contains $d(\pmb{s} , \pmb{s}')\ge 1$ qudits, is denoted by $B \subseteq C$, and the rest of $n-d(\pmb{s} , \pmb{s}')$ qudits, are denoted by $A \supseteq \overline{C}$.  Then, we can apply the exact same recursive construction as before. In particular, by properly labeling qudits in $A$ we obtain    $\mathbb{X}_C(\pmb{s};\pmb{s}')=2^{-l_C} F_{l_C}$, where $F_l$ is defined in Eq.(\ref{Fl}) and  $l_C=c-d(\pmb{s} , \pmb{s}')$  is the number of qudits in $C$ with the same reduced state   $|\pmb{s}\rangle$ and $|\pmb{s}'\rangle$  (This means we label qudits in $A$ such that the qudits in $A \cap C$ are labeled as  $1,\cdots, l_C$).  Then, applying  Eq.(\ref{eq26}) recursively, we can obtain the family of unitaries $\exp(\i\theta \mathbb{X}(\pmb{r};\pmb{r}')): \theta\in[0,2\pi)$, which proves lemma \ref{lem651}.

\subsection{Proof of corollary \ref{cor9}}

We saw how general elements of $\mathcal{SV}^G_{n,k}$ can be constructed from  $k$-local $G$-invariant unitaries. Next, we prove corollary \ref{cor9}, which highlights the fact that this construction only  requires specific type of $k$-local $G$-invariant unitaries.

To see this first  recall that  we obtained general unitaries 
in  $\mathcal{SV}^G_{n,k}$ by combining two families of unitaries  
\begin{itemize}
\item Single-qudit unitaries  $\exp(\frac{\i\pi }{2}|r\rangle\langle r|)$ with $|{r}\rangle\in\pmb{B}_1$.
\item  Unitaries $\exp(\i \theta \mathbb{X}(\textbf{r}; \textbf{r}'))$ for all pairs basis elements  $|\textbf{r}\rangle, |\textbf{r}'\rangle\in\pmb{B}_n$ that live in the same irreducible invariant subspace $\mathcal{H}_{\mu,\alpha}$. 
\end{itemize}
In particular, combing these families we obtain $\exp(\i \theta \mathbb{Y}(\textbf{r}; \textbf{r}'))$ and $\exp(\i \theta \mathbb{Z}(\textbf{r}; \textbf{r}'))$. Therefore, in the following we focus on realizing 
$\exp(\i \theta \mathbb{X}(\textbf{r}; \textbf{r}'))$ for arbitrary $|\textbf{r}\rangle, |\textbf{r}'\rangle\in\pmb{B}_n$ that belong to the same irreducible invariant subspace $\mathcal{H}_{\mu,\alpha}$. First, we show that for any such pair of basis elements,  there exists a sequence of basis elements $|\textbf{r}^{v}\rangle\in\pmb{B}_n: v=1,\cdots, w$, defined as 
 \bes\label{sec44}
 \begin{align}
 |\textbf{r}^{v+1}\rangle&= \mathbb{X}_{C_v}(\pmb{t}^{v+1} ; \pmb{t}^{v} )  |\textbf{r}^{v}\rangle\\ &=\mathbb{X}_{C_v}(\pmb{t}^{v+1} ; \pmb{t}^{v} ) 
  \cdots \cdots \mathbb{X}_{C_1}(\pmb{t}^{2} ; \pmb{t}^{1} ) |\pmb{r}\rangle\ ,
\end{align}
\ees
such that 
\begin{enumerate}
\item $|\textbf{r}^{1}\rangle=|\textbf{r}\rangle$ and $|\textbf{r}^{w}\rangle=|\textbf{r}'\rangle$. 
\item  $\mathbb{X}_{C_v}(\pmb{t}^{v+1} ; \pmb{t}^{v} )$ acts as $\mathbb{X}(\pmb{t}^{v+1} ; \pmb{t}^{v} )=|\pmb{t}^{v}\rangle\langle \pmb{t}^{v+1}|+|\pmb{t}^{v+1}\rangle\langle \pmb{t}^{v}|$ on qudits in $C_v$ and acts trivially on the rest of qudits, where $\mathbb{X}(\pmb{t}^{v+1} ; \pmb{t}^{v} )\in \textbf{H}_\text{redist}$. 
\end{enumerate}
To show the existence of such sequences, first recall that 
because  $|\textbf{r}\rangle, |\textbf{r}'\rangle\in\pmb{B}_n$ belong to the same irreducible invariant subspace 
$\mathcal{H}_{\mu,\alpha}$, according to lemma \ref{lem0}, there exists a sequence in the form of Eq.(\ref{seq}) such that any consecutive pair  $|\pmb{s}^{j}\rangle$ and  $|\pmb{s}^{j+1}\rangle$
have equal charges and Hamming distance $d(\pmb{s}^{j}, \pmb{s}^{j+1})\le k$.  The latter property means that to transform $|\pmb{s}^{j}\rangle$ to $|\pmb{s}^{j+1}\rangle$ it suffices to act on, at most, $k$ qudits, and  convert an element of    $\pmb{B}_k$ to another element of  $\pmb{B}_k$ with equal total charge. But,  according to the assumption of  corollary  \ref{cor9},  operators in $\textbf{H}_\text{redist}$
 act transitively on elements of  $\pmb{B}_k$ that have the same charge.  That is, we can convert $|\pmb{s}^{j}\rangle$ to $|\pmb{s}^{j+1}\rangle$ by applying a sequence of elements of $\textbf{H}_\text{redist}$. 
 This implies the existence of the  sequence defined in Eq.(\ref{sec44}) with the claimed  properties.

 Next, applying lemma \ref{lem651}, Eq.(\ref{sec44}) implies that for each $v=1,\cdots, w-1$,  the  unitaries 
 \be\label{rq}
 \exp(\i\theta \mathbb{X}(\textbf{r}^{v+1} ; \textbf{r}^{v})) \  : \theta\in[0,2\pi)\ ,
 \ee
  can be realized with by combining 
 2-qudit gates $C_{a,b}$ in Eq.(\ref{Cab}) with 
  unitaries $\exp(\i\theta \mathbb{X}(\pmb{t}^{v+1} ; \pmb{t}^{v} ))$, where $\mathbb{X}(\pmb{t}^{v+1} ; \pmb{t}^{v} )\in \textbf{H}_\text{redist}$. 

Finally, we note that because of the argument presented in  Eq.(\ref{commutq}),  by combining unitaries 
 in Eq.(\ref{rq})  with single-qudit unitaries  $\exp(\frac{\i\pi }{2}|r\rangle\langle r|)$ with $|{r}\rangle\in\pmb{B}_1$, we can realize the desired unitaries 
  $\exp(\i\theta \mathbb{X}(\textbf{r} ; \textbf{r}')) : \theta\in[0,2\pi) $. Note that this can be achieved for all pair 
 $|\textbf{r}\rangle, |\textbf{r}'\rangle\in\pmb{B}_n$ that belong to the same irreducible invariant subspace.

  In summary, we conclude that combining (i) single-qudit unitaries  $\exp(\frac{\i\pi }{2}|r\rangle\langle r|)$ with $|{r}\rangle\in\pmb{B}_1$, (ii)   2-qudit gates $C_{a,b}$ in Eq.(\ref{Cab}), and (iii) 
  unitaries $\exp(\i\theta \mathbb{X}(\pmb{t} ; \pmb{t}' ))$ with $\mathbb{X}(\pmb{t} ; \pmb{t}')\in \textbf{H}_\text{redist}$, we can realize any unitary in $\mathcal{SV}_{n,k}^G$. This completes the proof of corollary \ref{cor9}.

\subsection{Proof of theorem \ref{Thm4}}\label{secd}

Next, we show how theorem \ref{Thm4} follows from theorem \ref{Thm2}:

\noindent$\bf{ 1\Longleftrightarrow 2}$: Statement 1 in theorem \ref{Thm4} implies  that the group $\mathcal{V}^G_{n,k}$ acts irreducibly on subspaces $\{\mathcal{H}_\mu\}$.  This fact together with Schur's lemmas  immediately imply statement 2. Conversely, statement 2 implies  that $\mathcal{V}^G_{n,k}$ acts irreducibly on subspaces $\{\mathcal{H}_\mu\}$. This fact together with our theorem \ref{Thm2}, which is proved above, implies statement 1. Therefore, statements 1 and 2 are equivalent, which, in particular, prove theorem \ref{Thm3}.\\

 \noindent$\bf{ 2\Longrightarrow 3}$:  Statement 2 implies that  group $\mathcal{V}^G_{n,k}$ acts irreducibly on subspaces $\{\mathcal{H}_\mu\}$, which means each subspace $\mathcal{H}_\mu$
contains only a single irreducible invariant subspace. This, in particular, means   any pair of basis elements  $|\pmb{r}\rangle, |\pmb{r}'\rangle\in \pmb{B}_n$ that belong to the same charge sector $\mathcal{H}_\mu$, are in the same irreducible invariant subspace, and therefore the assumption of lemma 
 \ref{lem0} holds. Then, the lemma implies statement 3 of theorem \ref{Thm4}. \\

 \noindent$\bf{ 3\Longrightarrow 2}$:   Finally, we show that  statement 3 implies statement 2, and, hence statement 1. According to   
statement 3, for any pair of basis elements $|\pmb{r}\rangle$ and 
$|\pmb{r}'\rangle$ in the same charge sector $\mathcal{H}_\mu$, there exists a  sequence of basis elements $|\pmb{r}\rangle=|\pmb{s}^1\rangle, \cdots,  |\pmb{s}^t\rangle=|\pmb{r}'\rangle$, such that any consecutive pairs $|\pmb{s}^j\rangle$ and $|\pmb{s}^{j+1}\rangle$  are in the same charge sector $\mathcal{H}_\mu$ and have Hamming distance $d(\pmb{s}^{j}, \pmb{s}^{j+1})\le k$. Then, it is obvious that there exists a $k$-local symmetric unitary that converts $|\pmb{s}^{j}\rangle$ and $|\pmb{s}^{j+1}\rangle$, and therefore they are in the same irreducible invariant subspace (recall that any basis element lives in a single irreducible invariant subspace $\mathcal{H}_{\mu,\alpha}$). It follows that the basis elements $|\pmb{r}\rangle$ and $|\pmb{r}'\rangle$ also live in the same irreducible invariant subspace $\mathcal{H}_{\mu,\alpha}$. Since $|\pmb{r}\rangle$ and $|\pmb{r}'\rangle$ are arbitrary basis elements in  $\mathcal{H}_{\mu}$, this means $\mathcal{H}_{\mu}$  does not contain any proper non-trivial   irreducible invariant subspace. 
Therefore, $\text{Comm}(\mathcal{V}^G_{n, k})=\text{Span}_\mathbb{C}\{
\Pi_{\mu}\}$ which is the statement 2. This completes the proof of theorem \ref{Thm4}.\\

 \section*{Acknowledgments}
 I am grateful to Hanqing Liu for reading the manuscript and providing useful comments.  
 Also, I would like to thank Nicole Yunger Halpern for introducing me to her papers  on thermalization in the presence of non-Abelian conserved charges \cite{majidy2023non,  yunger2022build, kranzl2022experimental} and specifically  the conjecture in \cite{halpern2020noncommuting}. 
 
 This work is supported by a collaboration between the US
DOE and other Agencies. This material is based upon work
supported by the U.S. Department of Energy, Office of Science, National Quantum Information Science Research Centers, Quantum Systems Accelerator. Additional support is acknowledged from NSF Phy-2046195, NSF QLCI grant OMA-2120757, and ARL-ARO QCISS grant number W911NF-21-1-0005.

\bibliography{Ref_2021_v3, Ref_2020,Ref_2021}

\newpage

\onecolumngrid

\appendix
\newpage
\section{Constraints on the relative phases between sectors with different charges}\label{Sec:App1}

Here, we briefly review some relevant results of \cite{marvian2022restrictions} that allows us to characterize the constraints on the relative phases between sectors with different charges, and prove the results in section \ref{sec:typeI}.   Unless otherwise specified, in the following discussion the group $G$ can be any finite or compact Lie group, which can be Abelian or non-Abelian. 

For any Hamiltonian $H$  consider the function
\be
\chi_H(g)=\Tr( H U(g)) : \ \ g\in G\ ,
\ee
where $U(g)=u(g)^{\otimes n}: g\in G$ is the representation of group $G$ on the system. If $H$ is $G$-invariant, i.e., if $[H, U(g)]=0 : \forall g\in G$, then this function  is uniquely  determined by the real numbers
\be
\Tr(H \Pi_\mu): \mu\in \text{Irreps}_G(n)\ ,
\ee
which can be thought of as a vector in 
$\mathbb{R}^{|\text{Irreps}_G(n)|}$, and  in \cite{marvian2022restrictions}  is called the charge vector of $H$. In particular, 
\be\label{Fourier}
\chi_H(g)= \sum_{\mu\in \text{Irreps}_G(n) } \frac{1}{d_\mu}\Tr(\Pi_\mu H)\   f_\mu(g)\ , 
\ee
where $f_\mu$ is the character of irrep $\mu$ and $d_\mu$ is its dimension, which is equal to 1 in the case of Abelian groups (Recall that the character of a representation of $G$ is a complex function over group, defined as the trace of the representation).  

Eq.(\ref{Fourier}) is indeed the Fourier transform of the vector $\Tr(H \Pi_\mu): \mu\in \text{Irreps}_G(n)$, and is invertible via the inverse Fourier transform, namely 
\be
\Tr(\Pi_\mu H)=\frac{d_\mu}{|G|} \sum_{g\in G} \chi_H(g)\  f^\ast_\mu(g)\ ,
  \ee
  which follows from the orthogonality of characters (a similar relation holds for compact Lie groups, where the summation is replaced by the  integral with the Haar measure). In summary, function $\chi_H$ determines $\{\Tr(\Pi_\mu H)\}_\mu$, hence the component of $H$ in the subspace spanned by projectors $\{\Pi_\mu\}$. This subspace is indeed the  center of the Lie algebra of symmetric Hamiltonians, i.e., the Lie algebra associated to $\mathcal{V}^G_{n,n}$.  Note that the coefficients $\{\Tr(\Pi_\mu H)\}_\mu$ determine the relative phases of the unitaries realized by Hamiltonian $H$. In particular, for unitary  
 \be
 V=\exp(\i  t H)=\bigoplus_{\mu\in\text{Irreps}_G(n)} V_\mu\ ,
 \ee
  the determinant of the component of $V$  in the sector with irrep $\mu$ is
 \be
\text{det}(V_\mu)=\exp[\i t\  \Tr(\Pi_\mu H)]\ .
\ee 
    
    It can be shown that if Hamiltonian $H$ is realizable with $k$-local $G$-invariant unitaries, such that  
\be\label{tr2}
\forall t\in \mathbb{R}:\  \  \  \  \exp(\i H t)\in \mathcal{V}^G_{n, k}\ ,
\ee 
then,
\be
\chi_H= r^{n-k} \times \sum_{\mu\in \text{Irreps}_G(k) } c_\mu    f_\mu\ , 
\ee
for some real coefficients $c_\mu \in\mathbb{R}$, where  $r(g)=\Tr(u(g))$. Therefore, for such Hamiltonians we have
\be
\chi_H\in \text{Span}_\mathbb{R}\{r^{n-k} f_\nu:\ \nu\in  \text{Irreps}_G(k)\}\subseteq \text{Span}_\mathbb{R}\{f_\nu:\ \nu\in  \text{Irreps}_G(n)\}\ ,
\ee
where the second inclusion follows from Eq.(\ref{Fourier}). 
We conclude that for Hamiltonians that are realizable with $k$-local $G$-invariant unitaries, the vector 
$\Tr(H \Pi_\mu): \mu\in \text{Irreps}_G(n)$ can be any real vector in  $\mathbb{R}^{|\text{Irreps}_G(n)|}$ if, and only if 
\be\label{bb}
\text{Span}_\mathbb{R}\{r^{n-k} f_\nu:\ \nu\in  \text{Irreps}_G(k)\}=  \text{Span}_\mathbb{R}\{f_\nu:\ \nu\in  \text{Irreps}_G(n)\}\ .
\ee

More generally, the difference between the dimensions of these spaces, gives a lower bound on the difference between the dimensions of the Lie algebras associated to $\mathcal{V}_{n,n}^G$ and $\mathcal{V}_{n,k}^G$, i.e.,
\bes
\begin{align}
\text{dim}\big(\mathcal{V}_{n,n}^G\big)-\text{dim}\big(\mathcal{V}_{n,k}^G\big)&\ge \text{dim}\big(\text{Span}_\mathbb{R}\{f_\nu:\ \nu\in  \text{Irreps}_G(n)\}\big)-\text{dim}\big(\text{Span}_\mathbb{R}\{r^{n-k} f_\nu:\ \nu\in  \text{Irreps}_G(k)\}\big) \\ &=  |\text{Irreps}_G(n)|-\text{dim}\big(\text{Span}_\mathbb{R}\{r^{n-k} f_\nu:\ \nu\in  \text{Irreps}_G(k)\}\big)\\ &\ge |\text{Irreps}_G(n)|-|\text{Irreps}_G(k)|\ .\label{dh}
\end{align}
\ees
In certain cases of interest, the dimension of  $\text{Span}_\mathbb{R}\{r^{n-k} f_\nu:\ \nu\in  \text{Irreps}_G(k)\}$ is guaranteed to be equal  to $|\text{Irreps}_G(k)|$, which implies Eq.(\ref{dh}) holds  as equality.  In particular, suppose $r(g)=\Tr(u(g))\neq 0$ for all $g\in G$.  Then, for  real coefficients 
$c_\nu\in\mathbb{R}$, the  function $r^{n-k}  \sum_{\nu\in \text{Irreps}_G(k)} c_\nu f_\nu$ is equal to zero for all $g\in G$ if, and only if,  
$ \sum_{\nu\in \text{Irreps}_G(k)} c_\nu f_\nu$ is  equal to zero for all $g\in G$. But, the orthogonality of characters implies that this is possible only if all coefficients $c_\nu=0$. It follows that under this assumption, Eq.(\ref{dh}) holds as equality. Moreover, when $G$ is a connected Lie group, even if $r(g)=0$ for some group element $g\in G$, a similar argument can be made  based on the continuity of $r(g)$ around the identity element of the group. \\

On the other hand, if
there exists a group element $g_0\in G$ such that 
\be\label{cond}
r(g_0)=\Tr(u(g_0))=0\  
 , 
 \ee
  then, unless $k=n$, all functions in $\text{Span}_\mathbb{R}\{r^{n-k} f_\nu:\ \nu\in  \text{Irreps}_G(k)\}$, vanish at $g_0$.  But, if  $\text{Irreps}_G(n)$ contains a 1D irrep, which is always the case for  Abelian groups where all irreps are 1D,  $\text{Span}_\mathbb{R}\{f_\nu:\ \nu\in  \text{Irreps}_G(n)\}$ contains functions that are non-zero at $g_0$. This means for $k<n$,
   \be
\text{dim}\big(\text{Span}_\mathbb{R}\{r^{n-k} f_\nu:\ \nu\in  \text{Irreps}_G(k)\}\big)< \text{dim}(\text{Span}_\mathbb{R}\{f_\nu:\ \nu\in  \text{Irreps}_G(n)\})=  |\text{Irreps}_G(n)|\ ,
\ee
which, in turn, implies  the universality cannot be achieved with $k<n$. In particular, if $\mu\in \text{Irreps}_G(n)$ is a 1D irrep of $G$, then unless $k=n$, the family of unitaries $\exp(\i\theta \Pi_\mu): \theta\in[0,2\pi)$ cannot be realized with $k$-local $G$-invariant unitaries (except, for specific values of $\theta$).


In the following theorem, 
the commutator subgroup of $\mathcal{V}^G_{n,k}$,  is the subgroup generated by $V_1 V_2 V^\dag_1 V^\dag_2$ for $V_1, V_2\in \mathcal{V}^G_{n,k}$, i.e., 
\be
\left[\mathcal{V}^G_{n,k} , \mathcal{V}^G_{n,k} \right]=\big\langle V_1 V_2 V^\dag_1 V^\dag_2: V_1,V_2\in \mathcal{V}^G_{n,k}\big\rangle\ .
\ee
A general $G$-invariant unitary $V\in \mathcal{V}^G_{n,n}$ can be written as an element of the commutator subgroup $[\mathcal{V}^G_{n,n} , \mathcal{V}^G_{n,n} ]$ times a unitary in the subgroup 
\be
\big\{\sum_{\mu\in\text{Irreps}_G(n)} \exp(\i\phi_\mu)\Pi_\mu: \phi_\mu\in[0,2\pi)\big\} \ .
\ee
This subgroup corresponds to the center of the Lie algebra associated to the Lie group $\mathcal{V}^G_{n,n}$ (See \cite{marvian2022restrictions} for more detailed discussion).   Furthermore,
\be
\text{dim}\big(\mathcal{V}_{n,n}^G\big)=\text{dim}\big([\mathcal{V}_{n,n}^G, \mathcal{V}_{n,n}^G]\big)+|\text{Irreps}_G(n)|\ .
\ee

Based on the above arguments, Ref.  \cite{marvian2022restrictions} shows
  \begin{theorem}\emph{\cite{marvian2022restrictions}}\label{Thmsumm}
For any finite or compact Lie group $G$, it holds that 
\begin{align}
\text{dim}\big(\mathcal{V}_{n,n}^G\big)-\text{dim}\big(\mathcal{V}_{n,k}^G\big)&\ge|\text{Irreps}_G(n)|-\text{dim}\big(\text{Span}_\mathbb{R}\{r^{n-k} f_\nu:\ \nu\in  \text{Irreps}_G(k)\}\big)\label{ert}\\ &\ge |\text{Irreps}_G(n)|-|\text{Irreps}_G(k)|\ ,\label{ert2}
\end{align}
where $r(g)=\Tr(u(g))$. Furthermore,
\begin{itemize}
\item If the commutator subgroups of $\mathcal{V}_{n,n}^G$ and $\mathcal{V}_{n,k}^G$ are equal, i.e., 
$\big[ \mathcal{V}_{n,n}^G , \mathcal{V}_{n,n}^G  \big]=\big[ \mathcal{V}_{n,k}^G , \mathcal{V}_{n,k}^G  \big]$, then Eq.(\ref{ert}) holds as equality.
\item If $G$ is a connected Lie group, or $r(g)=\Tr(u(g))\neq 0$ for all $g\in G$,  then Eq.(\ref{ert2}) holds as equality.
\item Suppose there exists a group element $g_0\in G$ such that $\Tr(u(g_0))=0$. If $\text{Irreps}_G(n)$ contains a 1D irrep of $G$, which is always the case for Abelian groups,  then for $k<n$ the right-hand side of Eq.(\ref{ert}) is positive implying $\mathcal{V}_{n,k}^G\neq \mathcal{V}_{n,n}^G$.
\end{itemize}
\end{theorem}

For an Abelian group $G$, the group of all symmetric unitaries can be decomposed as  
\be
\mathcal{V}^G_{n,n}=\bigoplus_{\mu\in\text{Irreps}_G(n)} \hspace{-3mm}\text{U}(\mathcal{H}_{\mu})\ .
\ee
Then, using fact that the commutator subgroup of the unitary group $\text{U}(d)$ is the  special unitary group SU($d$), i.e., $[\text{U}(d), \text{U}(d)]=\text{SU}(d)$, we find that  the commutator subgroup of the group of all symmetric unitaries is
 \be\label{gk}
\left[\mathcal{V}^G_{n,n} , \mathcal{V}^G_{n,n} \right]=\hspace{-3mm}\bigoplus_{\mu\in\text{Irreps}_G(n)} \hspace{-3mm}\text{SU}(\mathcal{H}_{\mu})\ .
\ee 
Therefore, theorem  \ref{Thm3} implies that if there are no extra conserved observables, then 
\be
\left[\mathcal{V}^G_{n,k} , \mathcal{V}^G_{n,k} \right]=\left[\mathcal{V}^G_{n,n} , \mathcal{V}^G_{n,n} \right]=\bigoplus_{\mu\in\text{Irreps}_G(n)} \hspace{-3mm}\text{SU}(\mathcal{H}_{\mu})\ .
\ee
In this case  theorem \ref{Thmsumm} implies that
\begin{align}
\text{dim}\big(\mathcal{V}_{n,n}^G\big)-\text{dim}\big(\mathcal{V}_{n,k}^G\big)&=|\text{Irreps}_G(n)|-\text{dim}\big(\text{Span}_\mathbb{R}\{r^{n-k} f_\nu:\ \nu\in  \text{Irreps}_G(k)\}\big)\label{2ert}\\ &\ge |\text{Irreps}_G(n)|-|\text{Irreps}_G(k)|\ ,\label{2ert2}\ 
\end{align}
where Eq.(\ref{2ert2}) holds as equality if $G$ is connected, or $r(g)\neq 0$ for all $g\in G$. On the other hand, if $r(g)=0$ for some $g\in G$, then, unless $k=n$, the right-hand side of Eq.(\ref{2ert}) is strictly positive. This proves the results in section \ref{sec:typeI}.\\

Finally, we remark on the properties of $\text{Irreps}_G(l)$, i.e., the set of irreps of $G$ appearing in $u(g)^{\otimes l}: g\in G$.  Clearly, $\text{Irreps}_G(l+1)$ can be obtained by combining all irreps in $\text{Irreps}_G(l)$ with irreps in $\text{Irreps}_G(1)$. 

In the case of Abelian groups, the irreps are 1D, i.e., phases. Then, it is convenient to assume $\text{Irreps}_G(1)$ contains the trivial representation. This is always possible by multiplying  $u(g)$ in the inverse of a 1D irrep of $G$ that appear in $\text{Irreps}_G(1)$ (since this  is a global phase, this modification does not change the set of $G$-invariant operators).   With this convention, we find that for an Abelian group $G$, 
\be
\text{Irreps}_G(l) \subseteq \text{Irreps}_G(l+1)\ ,
\ee
and the elements of $\text{Irreps}_G(l+1)$ can be obtained by multiplying each element of  $\text{Irreps}_G(l)$ in an element of $\text{Irreps}_G(1)$.  This immediately implies that if for an integer  $l$, 
\be 
\text{Irreps}_G(l+1)=\text{Irreps}_G(l)\ \  \ \Longrightarrow\  \  \  \forall r>l:  \text{Irreps}_G(r)=\text{Irreps}_G(l)\ .   
\ee
For connected compact Abelian groups, such integer $l$ exists only if the representation of the symmetry on each qudit is trivial, i.e., $u(g)=\mathbb{I}_d$ for all $g\in G$ (This can be seen, for instance, by noting  that if  $\text{Irreps}_G(r)= \text{Irreps}_G(l)$, 
then for any irrep $f_\mu\in \text{Irreps}(1)$ and arbitrary $ r>l$, it holds that  $f^r_\mu\in \text{Irreps}_G(l)$. For a connected Lie group $G$,  consider the  first derivative(s) of $f^r_\mu$ in different directions over group $G$, at the identity element of the group. The magnitude(s) of these derivatives   grow linearly with $r$. Since $f^r_\mu\in \text{Irreps}_G(l)$ for all $r\ge l$, this is possible only if the derivative(s) of $f_\mu$ are zero at the identity element. For connected compact Lie groups, because the exponential map is surjective,  this implies that $f_\mu$ is constant, i.e., it is the trivial representation). 
 
In summary, we showed that
\begin{proposition}
For any compact Abelian group $G$
\be\label{theend}
|\text{Irreps}_G(l+1)|\ge |\text{Irreps}_G(l)|\ .
\ee
Furthermore, in the case of compact connected Abelian groups if $|\text{Irreps}_G(1)|>1$, which means the representation of group is non-trivial on each subsystem, then Eq.(\ref{theend}) is a strict inequality.  
\end{proposition}


\end{document}